\documentclass[12pt,a4paper]{article}
\usepackage[english]{babel}
\usepackage{float}
\usepackage{amsmath, amssymb, amsthm, epsfig, color, mathtools, bbm, dsfont, mathrsfs, bm}
\usepackage[latin1]{inputenc}
\usepackage[T1]{fontenc}
\usepackage{indentfirst}
\usepackage[shortlabels]{enumitem}
\usepackage[title]{appendix}
\usepackage{lmodern}       
\usepackage{setspace}
\usepackage{soul}
\usepackage{accents}
\usepackage{hyperref} 
\usepackage{geometry}
\usepackage{xcolor}
\usepackage{relsize}
\usepackage{scalerel}
\usepackage{tikz}

\usepackage{xr}    
\usepackage{graphicx} 
\usepackage[export]{adjustbox} 
\usepackage{subcaption} 
\usepackage{siunitx}
\usepackage{bm} 
\usepackage{comment}

\usepackage{addfont}
\addfont{OT1}{rsfs10}{\rsfs}
\usepackage{tikz, tkz-euclide}
\usetikzlibrary{patterns}
\usetikzlibrary{arrows.meta}
\usetikzlibrary{calc}
\usetikzlibrary{decorations.markings}
\usepackage{fancybox}
\newcommand{\be}{\begin{equation}}
\newcommand{\ee}{\end{equation}}

\newcommand{\I}{\mathrm{I}}

\numberwithin{equation}{section}

\newcounter{dummy} \numberwithin{dummy}{section}
  \theoremstyle{plain}
  \newtheorem*{theorem*}        {Theorem}
	\newtheorem*{conjecture*}   {Conjecture}
  \newtheorem{theorem}[dummy]          {Theorem}
  \newtheorem{lemma}[dummy]              {Lemma}
  \newtheorem*{lemma*}          {Lemma}
    
  \newtheorem{corollary}[dummy]           {Corollary}

  \newtheorem{remark}[dummy]           {Remark}
  
  \theoremstyle{remark}
  \theoremstyle{definition}
   \newtheorem{definition}[dummy]          {Definition}
\makeatletter

\makeatletter
\def\moverlay{\mathpalette\mov@rlay}
\def\mov@rlay#1#2{\leavevmode\vtop{%
		\baselineskip\z@skip \lineskiplimit-\maxdimen
		\ialign{\hfil$\m@th#1##$\hfil\cr#2\crcr}}}
\newcommand{\charfusion}[3][\mathord]{
	#1{\ifx#1\mathop\vphantom{#2}\fi
		\mathpalette\mov@rlay{#2\cr#3}
	}
	\ifx#1\mathop\expandafter\displaylimits\fi}
\makeatother

\newgeometry{vmargin={15mm}, hmargin={12mm,17mm}}
\allowdisplaybreaks
\begin{document}

\begin{center}
{\LARGE Phase Transitions on 1d Long-Range Ising Models with Decaying Fields: A Direct Proof via Contours}
\vskip.5cm
Lucas Affonso, Rodrigo Bissacot, Henrique Corsini, Kelvyn Welsch
\vskip.3cm
\begin{footnotesize}
Institute of Mathematics and Statistics (IME-USP), University of S\~{a}o Paulo, Brazil\\
\end{footnotesize}
\vskip.1cm
\begin{scriptsize}
emails: lucas.affonso.pereira@gmail.com, rodrigo.bissacot@gmail.com, henriquecorsini@gmail.com, kelvyn.emanuel@gmail.com
\end{scriptsize}
\end{center}

\begin{abstract}
      Following seminal work by J. Fr\"ohlich and T. Spencer on the critical exponent $\alpha=2$, we give a proof via contours of phase transition in the one-dimensional long-range ferromagnetic Ising model in the entire region of decay, where phase transition is known to occur, i.e., polynomial decay $\alpha \in (1,2]$. No assumptions that the nearest-neighbor interaction $J(1)$ is large are made. The robustness of the method also yields proof of phase transition in the presence of a nonsummable external field that decays sufficiently fast.
\end{abstract}

\section{Introduction}

In this paper, we revisit the powerful technique of multiscaled contours used by J. Fr\"ohlich and T. Spencer to prove phase transition for the one-dimensional ferromagnetic Ising model with interaction $J(r)=r^{-2}$. More precisely, we give a proof of phase transition by multiscaled contours to ferromagnetic one-dimensional Ising models, i.e., models with the formal Hamiltonian
\begin{align*}
    H_J(\sigma)=\sum_{x,y \in \mathbb{Z}} - J(\lvert x-y\rvert)\sigma_x\sigma_y,
\end{align*}
where $J(r)\geq 0$ and $\sigma_x\in\{-1,+1\}$, for all $x\in\mathbb{Z}$. For more details, see Section \ref{prelim}. While many of the techniques shown here may be applied to general interactions, we shall deal with $J(r)=r^{-\alpha}$, where $1<\alpha\leq 2$.

In \cite{Frohlich.Spencer.82}, Fr\"ohlich and Spencer were able to prove the existence of phase transition for the one-dimensional ferromagnetic Ising model with the critical decay exponent $\alpha=2$. They accomplished this by adapting the notion of decomposition by neutral multipoles, which they had developed the year before to give the first rigorous proof of BKT transition in \cite{FS81}, to the one-dimensional scenario. In \cite{Imbrie.82}, also published in 1982, introducing a very slight modification to the notion of multiscaled contours created by Fr\"ohlich and Spencer, Imbrie provided a low-temperature expansion to the one-dimensional ferromagnetic Ising model with decay $\alpha=2$ and was able to prove very sharp results regarding the two-point correlation function $\langle \sigma_x;\sigma_y\rangle_\beta$ in the regime of large $\beta$. 

In 2005, Cassandro, Ferrari, Merola, and Presutti revisited the notion of multiscaled contour introduced by Fr\"ohlich and Spencer and were able to prove the existence of phase transition for decaying exponents $\alpha\in [3-\log_23,2]$. They gave the contours a geometrical description as collections of triangles, whose size plays the role of energy estimators in their energy-entropy argument. To do so, they were forced to introduce the hypothesis that the nearest-neighbor interaction $J(1)$ be sufficiently large. Their geometrical approach involving the introduction of triangles, and the energy estimates on their sizes, were shown, by Littin and Picco in \cite{LittinPicco2017}, to be unable to capture the whole region of polynomial decay where phase transition is known, i.e., $\alpha\in (1,2]$. Despite the limitation, the idea of introducing a more geometrical approach allowed the community to obtain some additional results, such as phase separation and the convergence of the cluster expansion \cite{Cassandro.Merola.Picco.Rozikov.14}. Using triangles, Cassandro, Orlandi, and Picco proved, in \cite{Cassandro.Picco.09}, the robustness of phase transition in the presence of a random external field with small variance in the region $\alpha\in[3-\log_23,3/2)$. It is expected that such a result holds for $\alpha\in(1,3/2)$, which remains an open problem.  

In a previous paper \cite{Bissacot-Kimura2018}, Bissacot, Endo, van Enter, Kimura, and Ruszel managed to prove phase transition without the hypothesis that $J(1)\gg 1$ for $\alpha\in (\alpha^*,2]$, where $\zeta(\alpha^*)=2$. Note that $\alpha^*> 3-\log_23$. They were also able to prove the robustness of phase transition under perturbations by a small decaying field. The current paper extends the result to the entire region of polynomial decay, where phase transition is known to occur. We also obtain sharp results regarding phase transition and external fields with polynomial decay.

In \cite{Affonso.2021}, Affonso, Bissacot, Endo, and Handa were able to adapt the notion of multiscaled contours introduced by Fr\"ohlich and Spencer to the multidimensional case. This notion of contours was further refined in \cite{Johanes}, where it was applied to the long-range random field Ising model ($d\geq3$). They managed to capture the whole region of regularity of the interaction ($\alpha>d$), improving upon previous results by Park \cite{Park.88.I,Park.88.II} ($\alpha>3d+1$) and Ginibre, Grossmann and Ruelle \cite{Ginibre.Grossmann.Ruelle.66} ($\alpha>d+1$). The key idea in their construction of the multiscaled contours is the flexibility of the distancing exponent, denoted in their work by $a$, with respect to the decaying exponent $\alpha$. It is precisely this fact that made revisiting the one-dimensional long-range Ising model appealing. Interestingly, it turns out that the parameter $a$ is not crucial to our argument. By changing the way energy estimates are done, the exact same contour definition present in \cite{Frohlich.Spencer.82} is sufficient to prove phase transition for polynomial decay $\alpha\in (1,2]$. 

It is well-known \cite{Rogers1981} that there is no phase transition if
\begin{align}\label{criterion}
    \limsup_{n\rightarrow\infty}\frac{1}{(\log n)^{\frac{1}{2}}}\sum_{r=1}^nrJ(r)<\infty.
\end{align}
In which case, one might wonder if there is phase transition for interactions slightly weaker than $J(r)=r^{-2}$. Indeed, if $J(r)=r^{-2}[\log(1+\log(1+r))]^{-1}$, \eqref{criterion} cannot ascertain the lack of long-range order. Unfortunately, the method of multiscaled contours developed by Fr\"ohlich and Spencer more than forty years ago breaks down precisely at $J(r)=r^{-2}$, i.e., their energy-entropy arguments only works provided $r\mapsto J(r)$ is decreasing and 
\begin{align*}
    \liminf_{r\rightarrow\infty}r^2J(r)>0.
\end{align*}
On the other hand, if $J$ satisfies such a condition and decays polynomially, one can prove phase transition directly by multiscaled contours. Hence, direct proof is obtained via contours for the whole region of polynomial decay, where phase transition is known. This motivates revisiting important results of the one-dimensional theory, whose proofs rely on the hypothesis that $J(1)$ is large, see \cite{Cassandro.Merola.Picco.17, Cassandro.Merola.Picco.Rozikov.14, Cassandro.Picco.09}. We also emphasize that treating $J(1)$ as a perturbative parameter is not an idea confined solely to contour arguments. For example, in \cite{Aizenman1988}, phase transitions in the corresponding percolation models and the discontinuity of magnetization were rigorously established. These findings have recently been revisited and further developed by Duminil-Copin, Garban, and Tassion in \cite{DuminilCopin2024}, which also assumed the $J(1)\gg 1$ to obtain part of the results.

\subsection{Some observations on the multiscaled contour argument}

One must give a reason as to why such a powerful method breaks down precisely when $J(r)=r^{-2}$. It is so because at the critical polynomial decay, the natural object measuring energy coincides with the natural object measuring entropy, the notion of \textit{logarithmic length}. This fact gives the illusion that in \cite{Frohlich.Spencer.82} the entropy and energy estimates are tied from the start. This is but a coincidence. 

In which case, one wonders if it is possible to prove that the magnetization is always trivial for the one-dimensional long-range ferromagnetic Ising model if $r\mapsto J(r)$ is decreasing and 
\begin{align*}
    \liminf_{r\rightarrow\infty}r^2J(r)=0.
\end{align*}

Above the critical exponent ($1<\alpha<2$), it appears natural to try to extend the notion of logarithmic length to $J(r)=r^{-\alpha}$, for example, by the second integral of $r\mapsto J(r)$. This is the case in \cite{Cassandro.05}. Given their geometric interpretation of contours as collections of triangles, $\gamma=\{T_1,...,T_m\}$, and the expectation to produce a Pirogov-Sinai theory for this class of models, they arrive at an energy estimate of the type
\begin{align}\label{italians}
    H_J(\Gamma)-H_J(\Gamma\hspace{-0.1cm}\setminus\hspace{-0.1cm}\gamma)\geq C_\alpha\sum_{k}\lvert T_k\rvert^{2-\alpha} =: C_\alpha L_\alpha(\gamma),
\end{align}
where $L_\alpha(\gamma)$ is the \emph{$\alpha$-length} of the contour $\gamma$. However, to achieve such an estimate, they need to introduce the hypothesis that $J(1)$ be sufficiently large. Unfortunately, for $1<\alpha<2$, it is impossible to obtain an estimate of the type $c_\alpha L_\alpha(\gamma)\leq \mathcal{N}(\gamma)$, such as in Lemma 2.1 in \cite{Frohlich.Spencer.82} for $\alpha=2$, where $\mathcal{N}(\gamma)$ denotes the cover size (see Subsection \ref{scalingcovers}), which is the kind of connection needed to glue entropy and energy estimates. For this reason, their energy-estimate computation becomes much more involved. On the other hand, the triangle method was appropriate to prove \cite{Cassandro.Picco.09} phase transition in the presence of a random external field of sufficiently small variance in the region $(3-\log_23,3/2)$.

\subsection{Comparison with the multidimensional case}

In the multidimensional case, such as in \cite{Affonso.2021, Johanes, maia2024phase}, disregarding for a moment the fact that boundaries of configurations are given by incorrect points instead of spin flips, to prove the existence of phase transition, it is sufficient to obtain a bound of the type
\begin{align*}
    H_J(\Gamma)-H_J(\Gamma\hspace{-0.1cm}\setminus\hspace{-0.1cm}\gamma)\geq C_{M,a}\lvert\gamma\rvert,
\end{align*}
for some $C_{M,a}>0$, where $(M,a)$ are the parameters used to decompose the boundary of configurations into contours. This is the case because there is phase transition for the nearest-neighbor model in $d\geq 2$. It is only when perturbing the system by an external decaying field that contributions of the type 
\begin{align*}
    F_{\mathrm{I}_-(\gamma)}\coloneqq\sum_{\substack{x\in\mathrm{I}_-(\gamma)\\y\in \mathrm{I}_-(\gamma)^c}}J_{xy}=\frac{1}{2}H_J(\gamma),
\end{align*}
become necessary to prove the existence of phase transition, i.e., a bound of the type
\begin{align*}
    H_J(\Gamma)-H_J(\Gamma\hspace{-0.1cm}\setminus\hspace{-0.1cm}\gamma)\geq C_{M,a}(\lvert\gamma\rvert+F_{\mathrm{I}_-(\gamma)}),
\end{align*}
is needed. When dealing with the one-dimensional case, since there is no long-range order for the nearest-neighbors interaction, the quantity of spin flips $|\gamma|$ is not sufficient to give us a phase transition. One certainly needs an estimate of the second type, i.e., something like
\begin{align*}
    H_J(\Gamma)-H_J(\Gamma\hspace{-0.1cm}\setminus\hspace{-0.1cm}\gamma)\geq C_{M,a}H_J(\gamma).
\end{align*}
This is the case, for example, in \cite{Frohlich.Spencer.82}, although such an estimate is obtained somewhat indirectly.

\subsection{Outline of the proof: a standard Peierls' argument}\label{outline}

We follow the entropy-energy argument used by Fr\"ohlich and Spencer to prove phase transition for polynomial decay $\alpha=2$. However, our primary distinction lies in the method used to derive the energy estimates. In which case, we attempt to provide a revisiting as elementary as possible to the techniques present in \cite{Frohlich.Spencer.82}. 

We begin by rewriting configurations with homogeneous boundary conditions in terms of even collections of spin flips, i.e., we introduce a canonical application $\Omega^+\ni\sigma\mapsto\partial\sigma\in\Omega^*$. This canonical application is a bijection, which allows us to write without any ambiguity
\begin{align*}
    H_J(\sigma)=H_J(\partial\sigma).
\end{align*}
All details concerning such objects are given in Subsection \ref{spinflips}. These neutral collections of spin flips are to be partitioned in terms of two parameters $(M,a)$ (see Section \ref{contours}). Such a partition $\Gamma=\Gamma(\sigma,M,a)$ shall be called an \emph{$(M,a)$-partition} and its elements $\gamma\in\Omega^*$ irreducible contours. These partitions satisfy an ordering property, which allows us to define maximal elements, called {\it external contours}. We note that any contour $\gamma\in\Gamma$ shall be the image of a configuration $\sigma\in\Omega^+$. If $\Gamma\Subset\Omega^*$ and $\cup_{\gamma\in\Gamma}\gamma\in\Omega^*$, we shall write
\begin{align*}
    H_J(\Gamma)=H_J\left(\cup_{\gamma\in\Gamma}\gamma\right).
\end{align*}
Finally, in this language, if a configuration $\sigma\in\Omega^+$ satisfies $\sigma_x=-1$, then there is a contour $\gamma\in\Gamma(\sigma,M,a)$ such that $x\in V(\gamma)$, the {\it volume} of $\gamma$. All the precise definitions are given in the next subsections.

In general, let $\gamma\in\Gamma$ be an external contour, we shall obtain an energy bound of the type
\begin{align*}
    H_J(\Gamma)-H_J(\Gamma\hspace{-0.1cm}\setminus\hspace{-0.1cm}\gamma)\geq F(\gamma), 
\end{align*}
where $F(\gamma)$ is a function that associates to each $\gamma$ a real number which depends on the triple $(J, M, a)$. In our particular case, we get 
\begin{align}
    F(\gamma)=C_1H_J(\gamma),
\end{align}
where $C_1\coloneqq C_1(J,M,a)>0$ is a constant (see Subsection \ref{subsecenergyestimates}). It is precisely this fact that makes the proof work. This observation had already been done by Picco and Littin in \cite{LittinPicco2017}, although their work was focused in to show the limitations of the geometric approach with the triangles defined in \cite{Cassandro.05}. 

When dealing with the entropy part of the argument, we need a function $\mathcal{N}:\Omega^*\rightarrow\mathbb{N}_0 :=\mathbb{N}\cup \{0\}$ such that
\begin{align*}
    \left\lvert\left\{\gamma\in\Omega^*:\mathcal{N}(\gamma)\leq R, 0\in V(\gamma)\right\}\right\rvert\leq e^{C_2R},
\end{align*}
for some $C_2>0$ and all $R\in\mathbb{N}_0$. This function is the cover-size function present in \cite{Frohlich.Spencer.82}, which we revisit in detail in Subsection \ref{scalingcovers}.

Then, the energy-entropy argument works if, \emph{for some choice of $(M,a)$}, there exists a constant $C_3\coloneqq C_3(J,M,a)>0$ such that
\begin{align*}
    F(\gamma)\geq C_3\mathcal{N}(\gamma).
\end{align*}
This is the object of Subsections \ref{geometricway} and \ref{entropysizes}. If all these ingredients are available, we may prove phase transition via a Peierls' argument:
\begin{align*}
    \left\langle\frac{1}{2}(1-\sigma_0)\right\rangle^+_{J,\beta,L}&=\frac{\sum_\Gamma\chi_{\{\sigma_0=-1\}}e^{-\beta H_J(\Gamma)}}{\sum_\Gamma e^{-\beta H_J(\Gamma)}}\leq \frac{\sum_{\gamma:0\in V(\gamma)}\sum_{\Gamma\ni\gamma}e^{-\beta F(\gamma)}e^{-\beta H_J(\Gamma\setminus\gamma)}}{\sum_\Gamma e^{-\beta H_J(\Gamma)}}\\
    &\leq \sum_{\gamma:0\in V(\gamma)}e^{-\beta F(\gamma)}=\sum_{R=0}^\infty\sum_{\substack{\gamma:0\in V(\gamma)\\\mathcal{N}(\gamma)=R}} e^{-\beta F(\gamma)}\leq\sum_{R=0}^\infty\sum_{\substack{\gamma:0\in V(\gamma)\\\mathcal{N}(\gamma)=R}}e^{-\beta C_3 \mathcal{N}(\gamma)}\\
    &=\sum_{R=0}^\infty e^{-\beta C_33R}\lvert\{\gamma\in\Omega^*:\mathcal{N}(\gamma)=R,0\in V(\gamma)\}\rvert\\
    &\leq \sum_{R=0}^\infty e^{C_2R-\beta C_3R}\xrightarrow{\beta\rightarrow\infty}0.
\end{align*}
This means that 
\begin{align*}
    \lim_{\beta\rightarrow\infty}\left\langle\sigma_0\right\rangle^+_{J,\beta}=+1.
\end{align*}
Via the global spin-flip symmetry, we obtain that
\begin{align*}
    \lim_{\beta\rightarrow\infty}\langle\sigma_0\rangle^-_{J,\beta}=-1.
\end{align*}
This fact proves the existence of multiple Gibbs measures at low temperatures, one of the usual definitions of \textit{phase transition}, see \cite{georgii.gibbs.measures}. As a further application of irreducible contours, we show that such phase transition is stable under perturbation by a decaying magnetic field, see \ref{decaying_field}.

The paper is divided as follows. In Section \ref{prelim}, we provide the definitions regarding collections of spin flips. In Section \ref{contours}, we define the appropriate contour and establish appropriate energy and entropy estimates, which yield the existence of phase transition. In Section \ref{stability}, we prove the stability of phase transition under a small decaying field by a contour argument.

\section{Preliminaries}\label{prelim}

As usual, we denote the configuration space of a one-dimensional Ising model by $\Omega\coloneqq\{-1,1\}^\mathbb{Z}$,
endowed with the product topology. The set of configurations with $+$-boundary condition shall be denoted by
\begin{align}
    \Omega^+\coloneqq\{\sigma\in\Omega:\exists\Lambda\Subset\mathbb{Z}:\sigma_x=+1,\forall x\in\Lambda^c\},
\end{align}
where $A\Subset B$ denotes that $A$ is a finite subset of $B$. Given $\Lambda\Subset\mathbb{Z}$, we also define
\begin{align}
    \Omega_\Lambda^+\coloneqq\{\sigma\in\Omega^+:\sigma_x=+1, \forall x\in\Lambda^c\}.
\end{align}
In case $\Lambda=[-L,L]$, we shall write $\Omega^+_L$ instead.

By changing in a slight manner the usual way of writing the Hamiltonian of an Ising model, we shall not need to introduce explicitly the finite domains or boxes, over which one usually defines a family of compatible Hamiltonians. More precisely, we define $H_J:\Omega^+\rightarrow\mathbb{R}$ by
\begin{align}
    H_J(\sigma)=\sum_{x<y}J(\lvert x-y\rvert)(1-\sigma_x\sigma_y),
\end{align}
where $J:\mathbb{N}_0\rightarrow\mathbb{R}_{\geq0}$ is the non-negative function defined by $J(r)=r^{-\alpha}$ ($1<\alpha\leq 2$) when $r>0$ and $J(0)=0$. We are, therefore, dealing with a model defined by a translation-invariant {\it regular interaction}, see \cite{Bovier.06, FV-Book, georgii.gibbs.measures} for general results for regular models. We shall often write $J_{xy}$ instead of $J(\lvert x-y\rvert)$.

We introduce the usual finite-volume Gibbs measures at inverse temperature $\beta>0$, which we shall denote by $\langle\cdot\rangle_{J,\beta,L}^+$:
\begin{align}
    \langle f\rangle_{J,\beta,L}^+=\frac{1}{Z_{J,\beta,L}^+}\sum_{\sigma\in\Omega_L^+}f(\sigma)e^{-\beta H_J(\sigma)},
\end{align}
where $f:\Omega^+_L\rightarrow\mathbb{R}$ and
\begin{align}
    Z_{J,\beta,L}^+=\sum_{\sigma\in\Omega_L^+}e^{-\beta H_J(\sigma)},
\end{align}
is the {\it partition function}.

\subsection{Configurations as collections of spin flips}\label{spinflips}

We note that there is a one-to-one correspondence between configurations with homogeneous boundary conditions and even collections of spin flips, which are elements of the dual lattice $\mathbb{Z}^*$. As in \cite{Frohlich.Spencer.82}, we shall canonically identify it with $\mathbb{Z}+1/2$. Instead of working with the usual configuration space, it is more convenient to deal with such collections of spin flips in the one-dimensional case. We define
\begin{align}
    \Omega^*\coloneqq \{\gamma\in\mathcal{P}_f(\mathbb{Z}^*):\lvert\gamma\rvert\textrm{ is even}\},
\end{align}
where $\mathcal{P}_f(A)$ denotes the set of finite subsets of a set $A$. Define $\partial:\Omega^+\rightarrow\Omega^*$ by
\begin{align}
    \partial\sigma\coloneqq\{\{x,x+1\}\subset \mathbb{Z}:\sigma_x\sigma_{x+1}=-1\},
\end{align}
and identify $\{x,x+1\}\subset \mathbb{Z}$ with $x+1/2 \in \Omega^*$. We define the diameter $\mathrm{diam}(\gamma)$ of $\gamma\in\Omega^*$ as the diameter of $\gamma$ as a subset of $\mathbb{Z}+1/2$. It is evident that, if $\sigma$ and $\sigma'$ have the same boundary, i.e., if $\partial\sigma=\partial\sigma'$, then $\sigma=\sigma'$. We shall denote its inverse by $\sigma:\Omega^*\rightarrow\Omega^+$, i.e.,
\begin{align}
    \Omega^*\ni\gamma\mapsto\sigma(\gamma)\in\Omega^+.
\end{align}

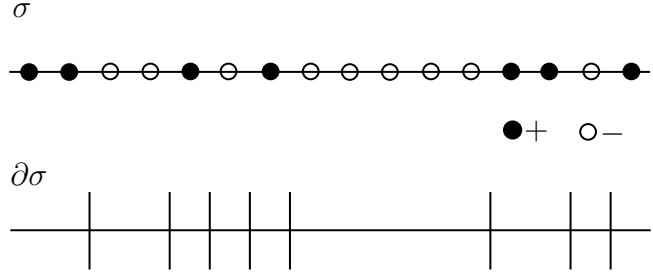
\begin{figure}[hbt!]
    \centering
    \tikzset{every picture/.style={line width=0.75pt}} 

\begin{tikzpicture}[x=0.75pt,y=0.75pt,yscale=-1,xscale=1]

\draw  [fill={rgb, 255:red, 0; green, 0; blue, 0 }  ,fill opacity=1 ] (25.93,90.4) .. controls (25.93,88.23) and (27.69,86.47) .. (29.87,86.47) .. controls (32.04,86.47) and (33.8,88.23) .. (33.8,90.4) .. controls (33.8,92.57) and (32.04,94.33) .. (29.87,94.33) .. controls (27.69,94.33) and (25.93,92.57) .. (25.93,90.4) -- cycle ;
\draw  [fill={rgb, 255:red, 0; green, 0; blue, 0 }  ,fill opacity=1 ] (45.93,90.4) .. controls (45.93,88.23) and (47.69,86.47) .. (49.87,86.47) .. controls (52.04,86.47) and (53.8,88.23) .. (53.8,90.4) .. controls (53.8,92.57) and (52.04,94.33) .. (49.87,94.33) .. controls (47.69,94.33) and (45.93,92.57) .. (45.93,90.4) -- cycle ;
\draw   (66.43,90.11) .. controls (66.43,87.94) and (68.19,86.18) .. (70.36,86.18) .. controls (72.53,86.18) and (74.29,87.94) .. (74.29,90.11) .. controls (74.29,92.29) and (72.53,94.05) .. (70.36,94.05) .. controls (68.19,94.05) and (66.43,92.29) .. (66.43,90.11) -- cycle ;
\draw   (86.43,90.11) .. controls (86.43,87.94) and (88.19,86.18) .. (90.36,86.18) .. controls (92.53,86.18) and (94.29,87.94) .. (94.29,90.11) .. controls (94.29,92.29) and (92.53,94.05) .. (90.36,94.05) .. controls (88.19,94.05) and (86.43,92.29) .. (86.43,90.11) -- cycle ;
\draw  [fill={rgb, 255:red, 0; green, 0; blue, 0 }  ,fill opacity=1 ] (106.43,90.11) .. controls (106.43,87.94) and (108.19,86.18) .. (110.36,86.18) .. controls (112.53,86.18) and (114.29,87.94) .. (114.29,90.11) .. controls (114.29,92.29) and (112.53,94.05) .. (110.36,94.05) .. controls (108.19,94.05) and (106.43,92.29) .. (106.43,90.11) -- cycle ;
\draw   (125.43,90.11) .. controls (125.43,87.94) and (127.19,86.18) .. (129.36,86.18) .. controls (131.53,86.18) and (133.29,87.94) .. (133.29,90.11) .. controls (133.29,92.29) and (131.53,94.05) .. (129.36,94.05) .. controls (127.19,94.05) and (125.43,92.29) .. (125.43,90.11) -- cycle ;
\draw  [color={rgb, 255:red, 0; green, 0; blue, 0 }  ,draw opacity=1 ][fill={rgb, 255:red, 0; green, 0; blue, 0 }  ,fill opacity=1 ] (146.43,90.11) .. controls (146.43,87.94) and (148.19,86.18) .. (150.36,86.18) .. controls (152.53,86.18) and (154.29,87.94) .. (154.29,90.11) .. controls (154.29,92.29) and (152.53,94.05) .. (150.36,94.05) .. controls (148.19,94.05) and (146.43,92.29) .. (146.43,90.11) -- cycle ;
\draw   (166.43,90.11) .. controls (166.43,87.94) and (168.19,86.18) .. (170.36,86.18) .. controls (172.53,86.18) and (174.29,87.94) .. (174.29,90.11) .. controls (174.29,92.29) and (172.53,94.05) .. (170.36,94.05) .. controls (168.19,94.05) and (166.43,92.29) .. (166.43,90.11) -- cycle ;
\draw   (185.93,90.4) .. controls (185.93,88.23) and (187.69,86.47) .. (189.87,86.47) .. controls (192.04,86.47) and (193.8,88.23) .. (193.8,90.4) .. controls (193.8,92.57) and (192.04,94.33) .. (189.87,94.33) .. controls (187.69,94.33) and (185.93,92.57) .. (185.93,90.4) -- cycle ;
\draw   (205.93,90.4) .. controls (205.93,88.23) and (207.69,86.47) .. (209.87,86.47) .. controls (212.04,86.47) and (213.8,88.23) .. (213.8,90.4) .. controls (213.8,92.57) and (212.04,94.33) .. (209.87,94.33) .. controls (207.69,94.33) and (205.93,92.57) .. (205.93,90.4) -- cycle ;
\draw   (226.43,90.11) .. controls (226.43,87.94) and (228.19,86.18) .. (230.36,86.18) .. controls (232.53,86.18) and (234.29,87.94) .. (234.29,90.11) .. controls (234.29,92.29) and (232.53,94.05) .. (230.36,94.05) .. controls (228.19,94.05) and (226.43,92.29) .. (226.43,90.11) -- cycle ;
\draw   (246.43,90.11) .. controls (246.43,87.94) and (248.19,86.18) .. (250.36,86.18) .. controls (252.53,86.18) and (254.29,87.94) .. (254.29,90.11) .. controls (254.29,92.29) and (252.53,94.05) .. (250.36,94.05) .. controls (248.19,94.05) and (246.43,92.29) .. (246.43,90.11) -- cycle ;
\draw  [fill={rgb, 255:red, 0; green, 0; blue, 0 }  ,fill opacity=1 ] (266.43,90.11) .. controls (266.43,87.94) and (268.19,86.18) .. (270.36,86.18) .. controls (272.53,86.18) and (274.29,87.94) .. (274.29,90.11) .. controls (274.29,92.29) and (272.53,94.05) .. (270.36,94.05) .. controls (268.19,94.05) and (266.43,92.29) .. (266.43,90.11) -- cycle ;
\draw  [color={rgb, 255:red, 0; green, 0; blue, 0 }  ,draw opacity=1 ][fill={rgb, 255:red, 0; green, 0; blue, 0 }  ,fill opacity=1 ] (285.43,90.11) .. controls (285.43,87.94) and (287.19,86.18) .. (289.36,86.18) .. controls (291.53,86.18) and (293.29,87.94) .. (293.29,90.11) .. controls (293.29,92.29) and (291.53,94.05) .. (289.36,94.05) .. controls (287.19,94.05) and (285.43,92.29) .. (285.43,90.11) -- cycle ;
\draw   (306.43,90.11) .. controls (306.43,87.94) and (308.19,86.18) .. (310.36,86.18) .. controls (312.53,86.18) and (314.29,87.94) .. (314.29,90.11) .. controls (314.29,92.29) and (312.53,94.05) .. (310.36,94.05) .. controls (308.19,94.05) and (306.43,92.29) .. (306.43,90.11) -- cycle ;
\draw  [color={rgb, 255:red, 0; green, 0; blue, 0 }  ,draw opacity=1 ][fill={rgb, 255:red, 0; green, 0; blue, 0 }  ,fill opacity=1 ] (326.43,90.11) .. controls (326.43,87.94) and (328.19,86.18) .. (330.36,86.18) .. controls (332.53,86.18) and (334.29,87.94) .. (334.29,90.11) .. controls (334.29,92.29) and (332.53,94.05) .. (330.36,94.05) .. controls (328.19,94.05) and (326.43,92.29) .. (326.43,90.11) -- cycle ;
\draw    (20,90.33) -- (339.67,90.33) ;
\draw  [fill={rgb, 255:red, 0; green, 0; blue, 0 }  ,fill opacity=1 ] (266.8,119.57) .. controls (266.8,117.19) and (268.73,115.27) .. (271.1,115.27) .. controls (273.47,115.27) and (275.4,117.19) .. (275.4,119.57) .. controls (275.4,121.94) and (273.47,123.87) .. (271.1,123.87) .. controls (268.73,123.87) and (266.8,121.94) .. (266.8,119.57) -- cycle ;
\draw   (304.83,120.51) .. controls (304.83,118.34) and (306.59,116.58) .. (308.76,116.58) .. controls (310.93,116.58) and (312.69,118.34) .. (312.69,120.51) .. controls (312.69,122.69) and (310.93,124.45) .. (308.76,124.45) .. controls (306.59,124.45) and (304.83,122.69) .. (304.83,120.51) -- cycle ;
\draw    (20.5,170) -- (340,170) ;
\draw    (60,150.75) -- (60,189.75) ;
\draw    (100,150.75) -- (100,189.75) ;
\draw    (120,150.75) -- (120,189.75) ;
\draw    (140,150.75) -- (140,189.75) ;
\draw    (160,150.75) -- (160,189.75) ;
\draw    (260,150.75) -- (260,189.75) ;
\draw    (300,150.75) -- (300,189.75) ;
\draw    (320,150.75) -- (320,189.75) ;

\draw (275.4,113.4) node [anchor=north west][inner sep=0.75pt]    {$+$};
\draw (313.2,114.8) node [anchor=north west][inner sep=0.75pt]    {$-$};
\draw (20,54.4) node [anchor=north west][inner sep=0.75pt]    {$\sigma $};
\draw (19,134.4) node [anchor=north west][inner sep=0.75pt]    {$\partial \sigma $};

\end{tikzpicture}
    \caption{Example of the correspondence between $\sigma\in\Omega^+$ and $\partial\sigma\in\Omega^*$.}
\end{figure}

The evenness of the boundary of a configuration with homogeneous boundary conditions is a consequence of the neutrality of the sum of all spin flips associated to it. If one were to work with the Potts model instead, such parity condition would be substituted by another one involving neutral sums over the finite group $\mathbb{Z}_q$. It just so happens that this arithmetic condition in the case of $q=2$ is a parity condition.

We define the Hamiltonian in the dual lattice by $H_J:\Omega^*\rightarrow\mathbb{R}$
\begin{align}
    H_J(\gamma)=2\sum_{x<y}J_{xy}\chi_{\gamma}(x,y),
\end{align}
where 
\begin{align}\label{chi}
    \chi_{\gamma}(x,y)=\begin{cases}
        0&\textrm{ if }\lvert [x,y]\cap\gamma\rvert\textrm{ is even,}\\
        1 &\textrm{ if }\lvert [x,y]\cap\gamma\rvert\textrm{ is odd.}
    \end{cases}
\end{align}
We note that the intervals $[x,y]$ in \eqref{chi} are \emph{real}. Finally, using the same notation is justified by the fact that $H_J(\partial\sigma)=H_J(\sigma)$, for all $\sigma\in\Omega^+$. Since we shall denote contours by $\gamma$, as a matter of easier reading, we write explicitly
\begin{align*}
    H_J(\gamma)=H_J(\sigma(\gamma)).
\end{align*}
Let $\gamma=\{b_1,...,b_{2m}\}\in\Omega^*$, with $b_i < b_j$ when $i<j$. We will always assume that each $\gamma \in \Omega^*$ is ordered according to the indices. We define
\begin{enumerate}
    \item its \textit{volume} by $V(\gamma)=[b_1,b_{2m}]\cap\mathbb{Z}$,
    \item its $-$\textit{-interior} by $\mathrm{I}_-(\gamma)=\sqcup_{i=1}^m\left([b_{2i-1},b_{2i}]\cap\mathbb{Z}\right)$,
    \item and its $+$\textit{-interior} by $\mathrm{I}_+(\gamma)=\sqcup_{i=1}^{m-1}\left([b_{2i},b_{2i+1}]\cap\mathbb{Z}\right)$.
\end{enumerate}
We have that $\sigma(\gamma)_x=-1$ if, and only if, $x\in \mathrm{I}_-(\gamma)$. Naturally, $V(\gamma)=\mathrm{I}_-(\gamma)\sqcup \mathrm{I}_+(\gamma)$. Finally, we note that
\begin{align}
    H_J(\sigma)=2F_{\mathrm{I}_-(\partial\sigma)}=\hspace{-0.3cm}\sum_{\substack{x\in \mathrm{I}_-(\partial\sigma)\\y\in \mathrm{I}_-(\partial\sigma)^c}}\hspace{-0.3cm}2J_{xy}=\hspace{-0.3cm}\sum_{\substack{x\in \mathrm{I}_-(\partial\sigma)\\y\in \mathrm{I}_+(\partial\sigma)}}\hspace{-0.2cm}2J_{xy}+\hspace{-0.3cm}\sum_{\substack{x\in \mathrm{I}_-(\partial\sigma)\\y\in V(\partial\sigma)^c}}\hspace{-0.2cm}2J_{xy}.
\end{align}

\subsubsection{Collections of collections of spin flips}

If $\Gamma\Subset\Omega^*$ is a finite collection of even collections of spin flips such that $\cup_{\gamma\in\Gamma}\gamma\in\Omega^*$, we define its \textit{volume} by
\begin{align*}
    V(\Gamma)=\bigcup_{\gamma\in\Gamma}V(\gamma).
\end{align*}
We define its \textit{negative part} by
\begin{align*}
    N(\Gamma)\coloneqq\left\{x\in\mathbb{Z}:\sigma(\Gamma)_x=-1\right\},
\end{align*}
where $\sigma(\Gamma)\in\Omega^+$ is the configuration associated to the even collection of spin flips $\cup_{\gamma\in\Gamma}\gamma$. In which case, we shall write
\begin{align}
    H_J(\Gamma)\coloneqq H_J\left(\cup_{\gamma\in\Gamma}\gamma\right)=H_J(\sigma(\Gamma)).
\end{align}
Note that $N(\Gamma)\subset V(\Gamma)$, in which case, we define its \textit{positive part} by $P(\Gamma)\coloneqq V(\Gamma)\setminus N(\Gamma)$.

In general, it is not at all evident how one may relate $N(\Gamma)$ to the collection of $\mathrm{I}_\pm(\gamma)$, for $\gamma\in\Gamma$. On the other hand, this may be done reasonably if the elements of $\Gamma$ satisfy an ordering property, which is true for the contour decomposition used in this work. Let $\gamma=\{a_1,...,a_m\}$ and $\gamma'=\{b_1,...,b_n\}$ be collections of spin flips. We shall write $\gamma<\gamma'$ if there exists $1\leq j\leq n-1$, such that
\begin{align}
    \gamma\subset (b_j,b_{j+1}).
\end{align}
Let $\Gamma\Subset\Omega^*$. We shall say that $\Gamma$ is \textit{well-ordered} if
\begin{align*}
    (\gamma\nleq \gamma')\wedge(\gamma'\nleq\gamma)\Rightarrow V(\gamma)\cap V(\gamma')=\emptyset.
\end{align*}

It is also important to define the notion of external elements of a well-ordered collection of even collections of spin flips. If $\Gamma$ is well-ordered, we define
\begin{align*}
    \mathcal{E}_{\mathrm{ext}}(\Gamma)\coloneqq \{\gamma\in\Gamma:\gamma\textrm{ is a maximal element of }\Gamma\}.
\end{align*}
In which case, we say that $\gamma\in\mathcal{E}_{\mathrm{ext}}(\Gamma)$ is \textit{external}. Note that
\begin{align*}
    V(\Gamma)=\hspace{-0.3cm}\bigsqcup_{\gamma\in\mathcal{E}_{\mathrm{ext}}(\Gamma)}\hspace{-0.2cm}V(\gamma).
\end{align*}
It is also easy to see that
\begin{align*}
    N\left(\mathcal{E}_{\mathrm{ext}}(\Gamma)\right)=\hspace{-0.3cm}\bigsqcup_{\gamma\in\mathcal{E}_{\mathrm{ext}}(\Gamma)}\hspace{-0.3cm}\mathrm{I}_-(\gamma).
\end{align*}

\begin{remark}
    Note that if $\Gamma$ is well-ordered, then $\cup_{\gamma\in\Gamma}\gamma\in\Omega^*$. Furthermore, if $\Gamma$ is well-ordered, that is the case for any subset of $\Gamma$.
\end{remark}

\begin{figure}[hbt!]
    \centering
    \input{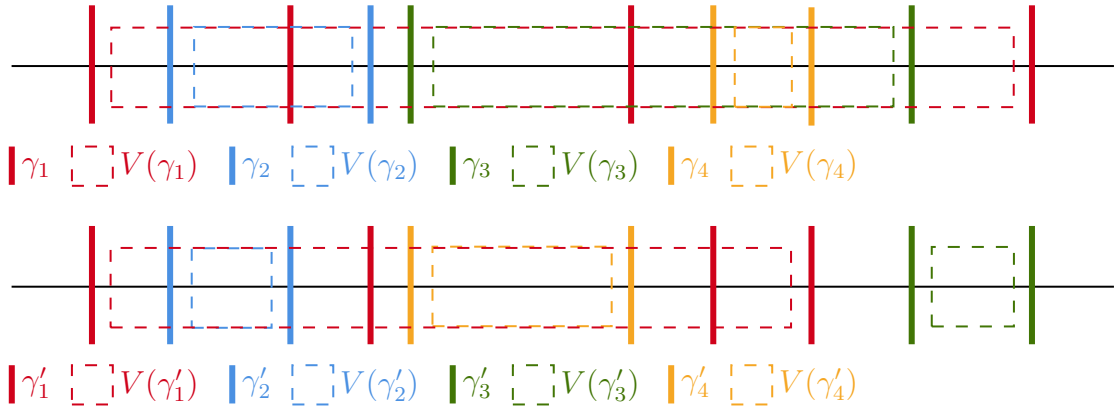}
    \caption{Let $\Gamma=\{\gamma_1,\gamma_2,\gamma_3,\gamma_4\}$. Then, $\Gamma$ is not well-ordered. On the other hand, $\Gamma'=\{\gamma'_1,\gamma'_2,\gamma'_3,\gamma'_4\}$ is well-ordered and $\mathcal{E}_{\mathrm{ext}}(\Gamma')=\{\gamma'_1,\gamma'_3\}$. While $\gamma_1\cup\gamma_2\cup\gamma_3\cup\gamma_4=\gamma'_1\cup\gamma'_2\cup\gamma'_3\cup\gamma'_4$, we have that $V(\Gamma)\neq V(\Gamma')$. Nevertheless, if $\cup_{\gamma\in\Gamma}\gamma=\cup_{\gamma'\in\Gamma'}\gamma'$, it is always true that $N(\Gamma)=N(\Gamma')$.}
\end{figure}

Let $\gamma\in\mathcal{E}_{\mathrm{ext}}(\Gamma)$. Define
\begin{align}
    \Gamma(\gamma)\coloneqq\{\gamma'\in\Gamma:\gamma'\leq \gamma\}.
\end{align}
Similarly, define
\begin{align}
    &\Gamma_\omega(\gamma)\coloneqq\{\gamma'\in\Gamma:\gamma'\subset \mathrm{I}_\omega(\gamma)\},
\end{align}
where $\omega\in\{-,+\}$. Note that $\Gamma(\gamma)=\{\gamma\}\sqcup\Gamma_+(\gamma)\sqcup\Gamma_-(\gamma)$.

\begin{lemma}\label{posneg}
    For a given collection $\Gamma$ of well-ordered spin-flips and $\gamma \in \mathcal{E}_{\mathrm{ext}}(\Gamma)$, it holds that
    \begin{enumerate}
        \item $N(\Gamma)=N(\Gamma\hspace{-0.1cm}\setminus\hspace{-0.1cm}\Gamma(\gamma))\sqcup N(\Gamma(\gamma))$,
        \item $N(\Gamma\hspace{-0.1cm}\setminus\hspace{-0.1cm}\gamma)=N(\Gamma\hspace{-0.1cm}\setminus\hspace{-0.1cm}\Gamma(\gamma))\sqcup N(\Gamma_+(\gamma))\sqcup N(\Gamma_-(\gamma))$, 
        \item $N(\Gamma(\gamma))=N(\Gamma_+(\gamma))\sqcup[\mathrm{I}_-(\gamma)\hspace{-0.1cm}\setminus\hspace{-0.1cm} N(\Gamma_-(\gamma))]$. 
    \end{enumerate}
\end{lemma}

\begin{proof}
    If $\Gamma$ is well-ordered, it holds that $N(\Gamma)=\bigsqcup_{\bar{\gamma}\in\mathcal{E}_{\textrm{ext}}(\Gamma)}N(\Gamma(\bar{\gamma}))$. In which case, item 1 is straightforward. Item 2 is a consequence of $\mathcal{E}_{\textrm{ext}}(\Gamma\setminus\gamma)=\left(\mathcal{E}_{\textrm{ext}}(\Gamma)\hspace{-0.1cm}\setminus\hspace{-0.1cm}\gamma\right)\sqcup\mathcal{E}_{\textrm{ext}}(\Gamma_+(\gamma))\sqcup\mathcal{E}_{\textrm{ext}}(\Gamma_-(\gamma))$.
    Finally, $\sigma(\Gamma_-(\gamma))_x=-\sigma(\Gamma(\gamma))_x$, if $x\in \mathrm{I}_-(\gamma)$. Conversely, $\sigma(\Gamma_+(\gamma))_x=\sigma(\Gamma(\gamma))_x$, if $x\in \mathrm{I}_+(\gamma)$. This is simply the observation that erasing $\gamma$ flips the spins in $\mathrm{I}_-(\gamma)$ and does nothing to the spins in $\mathrm{I}_-(\gamma)^c$. The lemma has, thus, been proved.
\end{proof}

\subsection{Scaling covers}\label{scalingcovers}

Following closely \cite{Frohlich.Spencer.82}, we define a collection of canonical open coverings of finite collections of spin-flips at different scales. The construction of such coverings shall be done recursively. Let $\{b_1,...,b_{2m}\}=\gamma\in\Omega^*$ be a finite collection of spin flips and $n\in\mathbb{N}_0$. We define 
$$I_1=\left(b_1-\frac{1}{2}, b_1-\frac{1}{2}+2^n\right).$$
Note that $b_1\in I_1$ and that the endpoints of $I_1$ are integers. Proceeding similarly with $\gamma_2\coloneqq\gamma\setminus I_1$ and following this procedure recursively, we arrive at a canonical minimal open covering of $\gamma$ by open intervals with diameter $2^n$ with integer endpoints. We shall denote such covering by $\mathcal{I}_n(\gamma)$.
\begin{remark}
    Note that $\mathcal{I}_n(\gamma)$ consists of a single interval if, and only if, $2^n>\mathrm{diam}(\gamma)$. Furthermore, $\lvert\mathcal{I}_0(\gamma)\rvert=\lvert\gamma\rvert$.
\end{remark}
\begin{remark}
    Let $I\in\mathcal{I}_0(\gamma)$. Then, $I\cap\mathbb{Z}$ is a translate of $\{0\}$. Let $I\in\mathcal{I}_n(\gamma)$, with $n\geq 1$. Then, $I\cap\mathbb{Z}=\{m\in\mathbb{Z}:1-2^{n-1}\leq m\leq 2^{n-1}-1\}$.
\end{remark}
\begin{remark}\label{endpoint}
    If $I\in\mathcal{I}_n(\gamma)$ and $I'\in\mathcal{I}_m(\gamma)$ are such that $I\cap\gamma=I'\cap\gamma$. Then, they have the same left endpoint.
\end{remark}

With this remark in mind, we define the cover size of $\gamma\in\Omega^*$ by
\begin{align}
    \mathcal{N}(\gamma)=\sum_{n=0}^{n_0(\gamma)}\lvert\mathcal{I}_n(\gamma)\rvert,
\end{align}
where $n_0(\gamma)=\lfloor\log_2\mathrm{diam}(\gamma)\rfloor$. Since this object is purely geometric and independent of the interaction, the exact same bound as stated in Theorem C in \cite{Frohlich.Spencer.82} is once again true. As a matter of completeness, we provide the proof for such an affirmation. 
\begin{theorem}\label{entropication}
    There is $C_2>0$ such that
    \begin{align}
            \lvert\mathcal{C}(R)\rvert\coloneqq \lvert\{\gamma\in\Omega^*:0\in V(\gamma), \mathcal{N}(\gamma)\leq R\}\rvert\leq e^{C_2R},
    \end{align}
    for all $R\in\mathbb{N}$.
\end{theorem}
\begin{proof}
Let $\gamma=\{b_1,...,b_{2m}\}$. If $n\leq\lfloor\log_2(b_2-b_1)\rfloor$, there is no open interval of diameter $2^n$ covering both $b_2$ and $b_1$. Hence,
\begin{align*}
    \mathcal{N}(\gamma)\geq\mathcal{N}(\gamma\setminus b_1)+\hspace{-0.5cm}\sum_{n=0}^{\lfloor\log_2(b_2-b_1)\rfloor}\hspace{-0.5 cm}1=\mathcal{N}(\gamma\setminus b_1)+1+\lfloor\log_2(b_2-b_1)\rfloor.
\end{align*}
Proceeding recursively, we arrive at
\begin{align*}
    \mathcal{N}(\gamma)\geq \sum_{j=1}^{\lvert\gamma\rvert-1}\left(1+\lfloor\log_2(b_{j+1}-b_j)\rfloor\right),
\end{align*}
one ought to recognize the right-hand side of the inequality above as the notion of logarithmic length present in \cite{Frohlich.Spencer.82}. It is, therefore, sufficient to prove that
\begin{align*}
    \lvert\mathcal{C}'(R)\rvert\coloneqq \left\lvert\left\{\gamma\in\Omega^*:0\in V(\gamma), \sum_{j=1}^{\lvert\gamma\rvert-1}(1+\lfloor\log_2(b_{j+1}-b_j)\rfloor)\leq R\right\}\right\rvert\leq e^{C_2R},
\end{align*}
for some $C_2>0$, for all $R\geq 1$. On the other hand, since $2(1+n_0(\gamma))\leq \mathcal{N}(\gamma)$, we have that
\begin{align*}
    \mathcal{N}(\gamma)\leq R\Longrightarrow n_0(\gamma)\leq \frac{R}{2}-1\Longrightarrow \mathrm{diam}(\gamma)\leq 2^{\frac{R}{2}}.
\end{align*}
Noting that any $\gamma\in\Omega^*$ is uniquely determined by the spin flip furthest to the left and the differences $(b_{j+1}-b_j)$ between consecutive spin flips, we obtain that
\begin{align*}
    \lvert\mathcal{C}(R)\rvert\leq \lvert\mathcal{C}'(R)\rvert\leq 2^{\frac{R}{2}}\sum_{n=1}^\infty\left\lvert\left\{(m_1,\dots,m_n)\in\mathbb{N}^n:\sum_{k=1}^n(1+\lfloor\log_2m_k\rfloor)\leq R \right\}\right\rvert.
\end{align*}
The $2^{\frac{R}{2}}$ factor coming from the possibilities of the spin-flip furthest to the left, given that $0\in V(\gamma)$, and the series being summed over the cardinality of $\lvert\gamma\rvert$.

In general, let $\varphi:\mathbb{N}\rightarrow\mathbb{N}$ be a function such that
\begin{align*}
    \lvert\{m\in\mathbb{N}:\varphi(m)=N\}\rvert\leq be^{aN},
\end{align*}
for all $N\in\mathbb{N}$ and some $a>0$, $b\geq1$. Then,
\begin{align*}
    \left\lvert\left\{(m_1,\dots,m_n)\in\mathbb{N}^n:\sum_{k=1}^n\varphi(m_k)\leq R\right\}\right\rvert&=\sum_{\substack{N\in\mathbb{N}^n\\N_1+...+N_n\leq R}}\prod_{k=1}^n|\{m\in \mathbb{N}:\varphi(m)= N_k\}|\\
    &\leq b^ne^{aR}\hspace{-0.6cm}\sum_{\substack{N\in\mathbb{N}^n\\N_1+...+N_n\leq R}}\hspace{-0.5cm}1=b^ne^{aR}\sum_{l=1}^R\binom{l-1}{n-1}.
\end{align*}
Hence,
\begin{align*}
    \sum_{n=1}^\infty\left\lvert\left\{m\in\mathbb{N}^n:\sum_{k=1}^n\varphi(m_k)\leq R\right\}\right\rvert&\leq e^{aR}\sum_{n=1}^\infty\sum_{l=1}^R\binom{l-1}{n-1}b^n\\
    &\leq e^{aR}b^R\sum_{l=1}^R2^{l-1}\leq e^{(a+\log (2b))R}.
\end{align*}
If $\varphi(m)=1+\lfloor\log_2m\rfloor$, we may take $a=\log 2$ and $b=1$. In which case, we arrive at
\begin{align}
    \lvert\mathcal{C}(R)\rvert\leq 2^{\frac{R}{2}}2^{2R}= e^{\left(\frac{5}{2}\log2\right)R}.
\end{align}
The exponential bound has, thus, been proved for $C_2=\frac{5}{2}\log2$.
\end{proof}
 
\section{Definition of long-range contours}\label{contours}

At first, it was somehow expected that following the strategy developed in \cite{Affonso.2021} would be necessary to prove phase transition for $1<\alpha<2$. More precisely, using incorrect points, not demanding neutrality of contours and changing the distancing parameter present in \cite{Frohlich.Spencer.82} from $3/2$ to some arbitrary $a$. However, by changing the way energy estimates were done, it turns out that \textit{any} distancing parameter $1<a<2$ is sufficient to capture the whole region of phase transition. That is, by changing the way the energy estimates are done, it becomes clear that Fr\"ohlich and Spencer had already provided most of the ideas necessary to capture the whole region of phase transition via a contour argument in \cite{Frohlich.Spencer.82}.

\begin{definition}\label{irreduciblecontour}
    Let $a\in\mathbb{R}$ be a real number such that $1<a<2$. Let $M\in\mathbb{R}$ be a real number such that $M>1$. Following the slight modifications introduced by Imbrie in \cite{Imbrie.82}, we say that a collection of spins $\gamma\in\Omega^*$ is an $(M,a)$-\textit{irreducible contour} if:
\begin{enumerate}
    \item the cardinality of $\gamma$ is even, i.e., it is neutrally charged;
    \item there is no decomposition of $\gamma$ into subsets $\gamma=\gamma_1\cup...\cup\gamma_n$ such that $\lvert\gamma_i\rvert$ is even for $1\leq i\leq n$, and $\mathrm{dist}(\gamma_i,\gamma_j)> M(\min\{\mathrm{diam}(\gamma_i),\mathrm{diam}(\gamma_j)\})^a$ for all $i\neq j$.
\end{enumerate}
\end{definition}

\begin{definition}
    Let $\partial\sigma\subset\mathbb{Z}^*$ be the collection of spin flips associated to $\sigma\in\Omega^+$. We say that $\Gamma\Subset\Omega^*$ is an $(M,a)$\textit{-partition} of $\partial\sigma$ if
\begin{enumerate}
    \item $\Gamma$ is a partition of $\partial\gamma$, i.e., $\gamma\cap\gamma'=\emptyset$, if $\gamma\neq\gamma'$ and $\gamma,\gamma'\in\Gamma$, and $\sqcup_{\gamma\in\Gamma}\gamma=\partial\sigma$;
    \item every element $\gamma\in\Gamma$ is an $(M,a)$-irreducible contour;
    \item $\mathrm{dist}(\gamma,\gamma')>M(\min\{\mathrm{diam}(\gamma),\mathrm{diam}(\gamma')\})^a$, whenever $\gamma,\gamma'\in\Gamma$ are distinct.
\end{enumerate}
\end{definition}

The uniqueness of the $(M,a)$-partition follows from the same argument given by Imbrie in Proposition 2.1 in \cite{Imbrie.82}.

\begin{lemma}
    Let $\sigma\in\Omega^+$. There is an unique $(M,a)$-partition $\Gamma\coloneqq\Gamma(\sigma,M,a)$ of $\partial\sigma$. 
\end{lemma}

\begin{definition}
    Given two $(M,a)$-irreducible contours $\gamma_1$ and $\gamma_2$, we say that they are \textit{compatible} if there exists $\sigma\in\Omega^+$ such that $\gamma_1$ and $\gamma_2$ are elements of $\Gamma(\sigma,M,a)$. In which case, we write $\gamma_1\simeq\gamma_2$.
    Similarly, let $\Gamma$ be an $(M,a)$-partition and $\gamma$ be an $(M,a)$-irreducible contour such that $\gamma\cap\gamma'=\emptyset$ for all $\gamma'\in\Gamma$. We say that $\gamma$ is compatible with $\Gamma$ if $\Gamma\sqcup\{\gamma\}=\Gamma(\sigma,M,a)$ for some $\sigma\in\Omega^+$. In which case, we write $\Gamma\simeq\gamma$.
\end{definition}

\begin{remark}
    Instead of writing $(M,a)$-irreducible contours, if no confusion can be made on $(M,a)$, we shall write irreducible contours instead.
\end{remark}

The fact that $a>1$ guarantees that compatible irreducible contours behave the same way they do in \cite{Frohlich.Spencer.82}. More precisely, given two distinct compatible irreducible contours $\gamma_1$ and $\gamma_2$ such that $V(\gamma_1)\cap V(\gamma_2)\neq \emptyset$. Then, either $\gamma_1<\gamma_2$ or $\gamma_2<\gamma_1$.

\begin{corollary}
    Let $\Gamma$ be the $(M,a)$-partition of $\partial\sigma$, where $\sigma\in\Omega^+$. Then, $\Gamma$ is well-ordered.
\end{corollary}

\begin{remark}
    In the multi-dimensional setting of previous works \cite{potts, Johanes, cluster, maia2024phase}, the property of being well-ordered was called \emph{\textbf{(A1)}}.
\end{remark}

\begin{definition}
    A maximal element of $\Gamma$ shall be called an \textit{exterior contour}.
\end{definition}

\subsection{Energy estimates}\label{subsecenergyestimates}

In this subsection, we establish a lower bound on the energy of removing an exterior contour $\gamma$ of an $(M,a)$-partition $\Gamma$. Since we are following the entropic approach developed in \cite{Frohlich.Spencer.82}, it is necessary to find energy estimates that are compatible with $a<2$. This warrants a brief discussion. 

When dealing with the interactions between $\gamma$ and large contours in the same way the estimates were done in \cite{Frohlich.Spencer.82}, this is the case of \cite{Affonso.2021} and \cite{maia2024phase}, one needs that
\begin{align*}
    1-a(\alpha-1)\leq0.
\end{align*}
This sort of condition arises from the following type of estimate, where it is expected to find $C>0$ such that, for $\Gamma''_1=\{\gamma''\in\Gamma:\mathrm{diam}(\gamma)\leq\mathrm{diam}(\gamma'')\}$, we have
\begin{align*}
    \sum_{\substack{x\in \mathrm{I}_-(\gamma)\\y\in V(\Gamma''_1)}}J_{xy}\leq 2\lvert \mathrm{I}_-(\gamma)\rvert\hspace{-0.4cm}\sum_{r=M\mathrm{diam}(\gamma)^a}^\infty \hspace{-0.4cm}J(r)\leq 2C\lvert \mathrm{I}_-(\gamma)\rvert\frac{1}{(M\mathrm{diam}(\gamma))^{a(\alpha-1)}}\leq \frac{2C}{M^{a(\alpha-1)}}\mathrm{diam}(\gamma)^{1-a(\alpha-1)}.
\end{align*}
Hoping to obtain an estimate as below, a condition on $a$ arises,
\begin{align*}
    \sum_{\substack{x\in \mathrm{I}_-(\gamma)\\y\in V(\Gamma''_1)}}J_{xy}\leq \frac{C}{M^{a(\alpha-1)}}.
\end{align*}
A similar phenomenon takes place in \cite{Affonso.2021} and \cite{maia2024phase}, which was overcome by losing the neutrality of the geometric objects. Given the geometry of the lattice in $d\geq 2$, the different geometric objects that arise are powerful tools to study multidimensional long-range Ising and Potts models as well, see \cite{potts}. Unfortunately, this is not the case for the one-dimensional Ising model. The idea to overcome such a difficulty came from an observation present in \cite{LittinPicco2017}, where a "relative" estimate is rather easy to obtain, as in an estimate of the type
\begin{align*}
    \sum_{\substack{x\in \mathrm{I}_-(\gamma)\\y\in V(\Gamma''_1)}}J_{xy}\leq C_{M,a}H_J(\gamma),
\end{align*}
instead of the previous "absolute" estimate. This may be done for any $1<a<2$, our strategy, therefore, follows closely the computations present in the Proposition 4.2.14 of \cite{maia2024phase} with a twist on the interaction between $\gamma$ and larger contours present in \cite{LittinPicco2017}.

\begin{lemma}\label{lemmabobo20}
    Let $J(r)=r^{-\alpha}$, where $1<\alpha\leq 2$. Let $n,M\in\mathbb{N}$. Then,
    \begin{align}
        \lim_{M\rightarrow\infty}\sup_{n\in\mathbb{N}}\frac{\sum_{r=nM}^\infty J(r)}{\sum_{r=n}^\infty J(r)}=\lim_{M\rightarrow\infty}\alpha M^{1-\alpha}=0.
    \end{align}
\end{lemma}

\begin{proof}
    Let $n\in\mathbb{N}$ and $M>1$. Then,
    \begin{align*}
        \sum_{r=nM}^\infty J(r)\leq (nM)^{-\alpha}+\int_{nM}^\infty\mathrm{d}rr^{-\alpha}\leq (nM)^{-\alpha}+(\alpha-1)^{-1}(nM)^{1-\alpha}.
    \end{align*}
    On the other hand,
    \begin{align*}
        \sum_{r=n}^\infty J(r)\geq \int_n^\infty\mathrm{d}rr^{-\alpha}=(\alpha-1)^{-1}n^{1-\alpha}.
    \end{align*}
    Therefore,
    \begin{align*}
        \frac{\sum_{r=nM}^\infty J(r)}{\sum_{r=n}^\infty J(r)}\leq \frac{\left[(\alpha-1)(nM)^{-1}+1\right](nM)^{1-\alpha}}{n^{1-\alpha}}\leq\alpha M^{1-\alpha}.
    \end{align*}
    The lemma has, thus been proved.
\end{proof}

\begin{lemma}\label{energyestimates}
    Let $\Gamma\coloneqq \Gamma(\sigma,M,a)$ and $\gamma\in\mathcal{E}_{\mathrm{ext}}(\Gamma)$. Then, there exists $C=C(a)>0$ such that
    \begin{align}
        H_J(\Gamma)-H_J(\Gamma\hspace{-0.1cm}\setminus\hspace{-0.1cm}\gamma)\geq \left(1-\frac{C}{M}-4\alpha M^{1-\alpha}\right)H_J(\gamma).
    \end{align}
\end{lemma}


\begin{proof}
    In this proof, we shall denote the configuration associated to $\Gamma\hspace{-0.1cm}\setminus\hspace{-0.1cm}\gamma$ by $\tau$. Then,
    \begin{align*}
        H_J(\Gamma)-H_J(\Gamma\hspace{-0.1cm}\setminus\hspace{-0.1cm}\gamma)&=\frac{1}{2}\sum_{x,y}J_{xy}(\tau_x\tau_y-\sigma_x\sigma_y)=\frac{1}{2}\sum_{x,y\in V(\gamma)}J_{xy}(\tau_x\tau_y-\sigma_x\sigma_y)+\sum_{\substack{x\in V(\gamma)\\y\in V(\gamma)^c}}J_{xy}(\tau_x-\sigma_x)\sigma_y\\
        &=\sum_{\substack{x\in \mathrm{I}_-(\gamma)\\y\in \mathrm{I}_+(\gamma)}}J_{xy}(-2\sigma_x\sigma_y)+\sum_{\substack{x\in \mathrm{I}_-(\gamma)\\y\in V(\gamma)^c}}J_{xy}(-2\sigma_x)\sigma_y=-\sum_{\substack{x\in \mathrm{I}_-(\gamma)\\y\in \mathrm{I}_-(\gamma)^c}}2J_{xy}\sigma_x\sigma_y.
    \end{align*}
    The second equality follows from the fact that $\tau_x=\sigma_x$ if $x\in V(\gamma)^c$, since $\gamma$ is an external contour. The third equality follows from the observation that $\tau_x=\sigma_x$, if $x\in \mathrm{I}_+(\gamma)$. Similarly, $\tau_x=-\sigma_x$, if $x\in \mathrm{I}_-(\gamma)$.

    To proceed with the computations, we partition $\Gamma\hspace{-0.1cm}\setminus\hspace{-0.1cm}\gamma$ in two parts. We define
    \begin{align*}
        \Gamma'\coloneqq \{\gamma'\in \Gamma\hspace{-0.1cm}\setminus\hspace{-0.1cm}\gamma: V(\gamma')\subset \mathrm{I}_-(\gamma)\}=\Gamma_-(\gamma)
    \end{align*}
    and $\Gamma''=\Gamma\setminus(\gamma\cup\Gamma')$. Note that $\sigma_x=-1$, if $x\in \mathrm{I}_-(\gamma)\setminus V(\Gamma')$. Similarly, $\sigma_x=+1$, if $x\in \mathrm{I}_-(\gamma)^c\setminus V(\Gamma'')$. Once again, the fact that $\gamma$ is an exterior contour guarantees that $V(\Gamma'')\subset \mathrm{I}_-(\gamma)^c$. We obtain
    \begin{align*}
        H_J(\Gamma)-H_J(\Gamma\hspace{-0.1cm}\setminus\hspace{-0.1cm}\gamma)=&-\hspace{-0.2cm}\sum_{\substack{x\in V(\Gamma')\\y\in V(\Gamma'')}}\hspace{-0.2 cm}2J_{xy}\sigma_x\sigma_y+\hspace{-0.4 cm}\sum_{\substack{x\in \mathrm{I}_-(\gamma)\setminus V(\Gamma')\\y\in V(\Gamma'')}}\hspace{-0.5cm}2J_{xy}\sigma_y-\hspace{-0.4cm}\sum_{\substack{x\in V(\Gamma')\\y\in \mathrm{I}_-(\gamma)^c\setminus V(\Gamma'')}}\hspace{-0.5cm}2J_{xy}\sigma_x+\hspace{-0.4cm}\sum_{\substack{x\in \mathrm{I}_-(\gamma)\setminus V(\Gamma')\\\mathrm{I}_-(\gamma)^c\setminus V(\Gamma'')}}\hspace{-0.4cm}2J_{xy}.
    \end{align*}
    Since the spin space is simply $\{-1,+1\}$, we have the following two identities
    \begin{align*}
        &\sigma_x\sigma_y=1-2\mathbf{1}_{\sigma_x\neq\sigma_y}\\
        &\sigma_x=1-2\mathbf{1}_{\sigma_x=-1}=2\mathbf{1}_{\sigma_x=+1}-1.
    \end{align*}
    Therefore,
    \begin{align*}
        H_J(\Gamma)-H_J(\Gamma\hspace{-0.1cm}\setminus\hspace{-0.1cm}\gamma)=&-\hspace{-0.2cm}\sum_{\substack{x\in V(\Gamma')\\y\in V(\Gamma'')}}2J_{xy}+\hspace{-0.4cm}\sum_{\substack{x\in \mathrm{I}_-(\gamma)\setminus V(\Gamma')\\y\in V(\Gamma'')}}\hspace{-0.4cm}2J_{xy}+\hspace{-0.4cm}\sum_{\substack{x\in V(\Gamma')\\y\in \mathrm{I}_-(\gamma)^c\setminus V(\Gamma'')}}\hspace{-0.5cm}2J_{xy}+\hspace{-0.4cm}\sum_{\substack{x\in \mathrm{I}_-(\gamma)\setminus V(\Gamma')\\ y\in \mathrm{I}_-(\gamma)^c\setminus V(\Gamma'')}}\hspace{-0.4cm}2J_{xy}\\
        &+\hspace{-0.2cm}\sum_{\substack{x\in V(\Gamma')\\y\in V(\Gamma'')}}4J_{xy}\mathbf{1}_{\sigma_x\neq\sigma_y}-\hspace{-0.4cm}\sum_{\substack{x\in \mathrm{I}_-(\gamma)\setminus V(\Gamma')\\ y\in V(\Gamma'')}}\hspace{-0.4cm}4J_{xy}\mathbf{1}_{\sigma_y=-1}-\hspace{-0.4cm}\sum_{\substack{x\in V(\Gamma')\\y\in \mathrm{I}_-(\gamma)\setminus V(\Gamma'')}}\hspace{-0.4cm}4J_{xy}\mathbf{1}_{\sigma_x=+1}.
    \end{align*}
    Recalling that $H_J(\gamma)=2\sum_{\substack{x\in \mathrm{I}_-(\gamma)\\y\in \mathrm{I}_-(\gamma)^c}}J_{xy}$, we obtain that
    \begin{align*}
        H_J(\Gamma)-H_J(\Gamma\hspace{-0.1cm}\setminus\hspace{-0.1cm}\gamma)&\geq H_J(\gamma)-\sum_{\substack{x\in V(\Gamma')\\y\in V(\Gamma'')}}4J_{xy}-\sum_{\substack{x\in \mathrm{I}_-(\gamma)\setminus V(\Gamma')\\y\in V(\Gamma'')}}4J_{xy}\mathbf{1}_{\sigma_y=-1}-\sum_{\substack{x\in V(\Gamma')\\ y\in \mathrm{I}_-(\gamma)^c\setminus V(\Gamma'')}}4J_{xy}\mathbf{1}_{\sigma_x=+1}\\
        &\geq H_J(\gamma) -\sum_{\substack{x\in V(\Gamma')\\y\in \mathrm{I}_-(\gamma)^c}}4J_{xy}-\sum_{\substack{x\in \mathrm{I}_-(\gamma)\\y\in V(\Gamma'')}}4J_{xy}.
    \end{align*}
    
    Let us begin by dealing with the sum involving $V(\Gamma'')$. We define $\Gamma''_1\coloneqq \{\gamma''\in\Gamma'':\mathrm{diam}(\gamma)\leq \mathrm{diam}(\gamma'')\}$ and $\Gamma''_2=\Gamma''\setminus\Gamma''_1$. It is precisely here that one ought to suppose that $J$ decays polynomially, such a computation does not work, for example, if $J(r)=r^{-1}(\log(1+(\log(1+x))))^{-2}$.
    
    Since $\gamma$ is an exterior contour, $V(\gamma)\cap V(\gamma'')=\emptyset$, if $\gamma''\in\Gamma''_1$. Hence,
    \begin{align*}
        \sum_{\substack{x\in \mathrm{I}_-(\gamma)\\y\in V(\Gamma''_1)}}J_{xy}&\leq \sum_{\substack{x\in \mathrm{I}_-(\gamma)\\\mathrm{dist}(y, V(\gamma))\geq M\mathrm{diam}(\gamma)^a}}J_{xy}\leq \sum_{\substack{x\in \mathrm{I}_-(\gamma)\\\mathrm{diam}(y, V(\gamma))>M\mathrm{diam}(\gamma)}}J_{xy}\\
        &\leq \sup_{\substack{\gamma\in\Omega^* \\x\in V(\gamma)}}\left[\frac{\sum_{\mathrm{diam}(y,V(\gamma))>M\mathrm{diam}(\gamma)}J_{xy}}{\sum_{y\in V(\gamma)^c}J_{xy}}\right]\sum_{\substack{\mathrm{I}_-(\gamma)\\ y\in V(\gamma)^c}}J_{xy}\leq \alpha M^{1-\alpha}H_J(\gamma),
    \end{align*}
    where $\lim_{M\rightarrow\infty}\alpha M^{1-\alpha}=0$. The last inequality follows from Lemma \eqref{lemmabobo20}.
    
    To proceed with the computations, we further partition $\Gamma''_2$. For $r\in\mathbb{N}_0$, we define
    \begin{align*}
        \Gamma''_{2,r}=\{\gamma''\in\Gamma''_2:2^r\leq \mathrm{diam}(\gamma'')<2^{r+1}\}.
    \end{align*}
    Note that $\mathrm{dist}(\gamma''_1,\gamma''_2)\geq M2^{ar}$ and that $V(\gamma''_1)\cap V(\gamma''_2)=\emptyset$, whenever $\gamma''_1,\gamma''_2\in \Gamma''_{2,r}$. 

    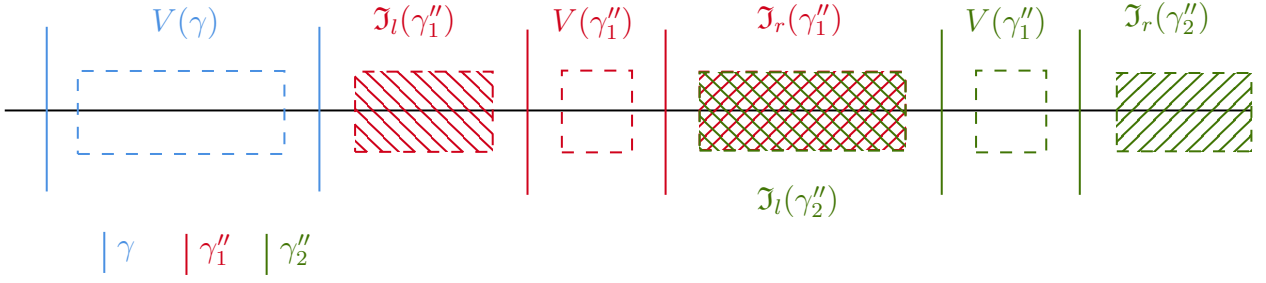
\begin{figure}[hbt!]
        \centering
 
\tikzset{
pattern size/.store in=\mcSize, 
pattern size = 5pt,
pattern thickness/.store in=\mcThickness, 
pattern thickness = 0.3pt,
pattern radius/.store in=\mcRadius, 
pattern radius = 1pt}
\makeatletter
\pgfutil@ifundefined{pgf@pattern@name@_h7e63vsv8}{
\pgfdeclarepatternformonly[\mcThickness,\mcSize]{_h7e63vsv8}
{\pgfqpoint{0pt}{-\mcThickness}}
{\pgfpoint{\mcSize}{\mcSize}}
{\pgfpoint{\mcSize}{\mcSize}}
{
\pgfsetcolor{\tikz@pattern@color}
\pgfsetlinewidth{\mcThickness}
\pgfpathmoveto{\pgfqpoint{0pt}{\mcSize}}
\pgfpathlineto{\pgfpoint{\mcSize+\mcThickness}{-\mcThickness}}
\pgfusepath{stroke}
}}
\makeatother

 
\tikzset{
pattern size/.store in=\mcSize, 
pattern size = 5pt,
pattern thickness/.store in=\mcThickness, 
pattern thickness = 0.3pt,
pattern radius/.store in=\mcRadius, 
pattern radius = 1pt}
\makeatletter
\pgfutil@ifundefined{pgf@pattern@name@_tuf4mjvvs}{
\pgfdeclarepatternformonly[\mcThickness,\mcSize]{_tuf4mjvvs}
{\pgfqpoint{0pt}{0pt}}
{\pgfpoint{\mcSize+\mcThickness}{\mcSize+\mcThickness}}
{\pgfpoint{\mcSize}{\mcSize}}
{
\pgfsetcolor{\tikz@pattern@color}
\pgfsetlinewidth{\mcThickness}
\pgfpathmoveto{\pgfqpoint{0pt}{0pt}}
\pgfpathlineto{\pgfpoint{\mcSize+\mcThickness}{\mcSize+\mcThickness}}
\pgfusepath{stroke}
}}
\makeatother

 
\tikzset{
pattern size/.store in=\mcSize, 
pattern size = 5pt,
pattern thickness/.store in=\mcThickness, 
pattern thickness = 0.3pt,
pattern radius/.store in=\mcRadius, 
pattern radius = 1pt}
\makeatletter
\pgfutil@ifundefined{pgf@pattern@name@_ywf33vic5}{
\pgfdeclarepatternformonly[\mcThickness,\mcSize]{_ywf33vic5}
{\pgfqpoint{0pt}{0pt}}
{\pgfpoint{\mcSize+\mcThickness}{\mcSize+\mcThickness}}
{\pgfpoint{\mcSize}{\mcSize}}
{
\pgfsetcolor{\tikz@pattern@color}
\pgfsetlinewidth{\mcThickness}
\pgfpathmoveto{\pgfqpoint{0pt}{0pt}}
\pgfpathlineto{\pgfpoint{\mcSize+\mcThickness}{\mcSize+\mcThickness}}
\pgfusepath{stroke}
}}
\makeatother

 
\tikzset{
pattern size/.store in=\mcSize, 
pattern size = 5pt,
pattern thickness/.store in=\mcThickness, 
pattern thickness = 0.3pt,
pattern radius/.store in=\mcRadius, 
pattern radius = 1pt}
\makeatletter
\pgfutil@ifundefined{pgf@pattern@name@_fv2ow3wzu}{
\pgfdeclarepatternformonly[\mcThickness,\mcSize]{_fv2ow3wzu}
{\pgfqpoint{0pt}{-\mcThickness}}
{\pgfpoint{\mcSize}{\mcSize}}
{\pgfpoint{\mcSize}{\mcSize}}
{
\pgfsetcolor{\tikz@pattern@color}
\pgfsetlinewidth{\mcThickness}
\pgfpathmoveto{\pgfqpoint{0pt}{\mcSize}}
\pgfpathlineto{\pgfpoint{\mcSize+\mcThickness}{-\mcThickness}}
\pgfusepath{stroke}
}}
\makeatother
\tikzset{every picture/.style={line width=0.75pt}} 

\begin{tikzpicture}[x=0.75pt,y=0.75pt,yscale=-1,xscale=1]

\draw    (10.33,96.59) -- (632.33,96.59) ;
\draw [color={rgb, 255:red, 74; green, 144; blue, 226 }  ,draw opacity=1 ]   (31.28,54.5) -- (31.28,137.34) ;
\draw [color={rgb, 255:red, 74; green, 144; blue, 226 }  ,draw opacity=1 ]   (167.33,54.5) -- (167.33,137.34) ;
\draw  [color={rgb, 255:red, 74; green, 144; blue, 226 }  ,draw opacity=1 ][dash pattern={on 4.5pt off 4.5pt}] (46.78,76.73) -- (149.87,76.73) -- (149.87,118.48) -- (46.78,118.48) -- cycle ;
\draw [color={rgb, 255:red, 208; green, 2; blue, 27 }  ,draw opacity=1 ]   (271.16,55.9) -- (271.16,139.03) ;
\draw [color={rgb, 255:red, 208; green, 2; blue, 27 }  ,draw opacity=1 ]   (340.05,55.9) -- (340.05,139.03) ;
\draw [color={rgb, 255:red, 65; green, 117; blue, 5 }  ,draw opacity=1 ]   (477.83,55.9) -- (477.83,139.03) ;
\draw [color={rgb, 255:red, 65; green, 117; blue, 5 }  ,draw opacity=1 ]   (546.71,55.9) -- (546.71,139.03) ;
\draw  [color={rgb, 255:red, 208; green, 2; blue, 27 }  ,draw opacity=1 ][dash pattern={on 4.5pt off 4.5pt}] (288.29,76.54) -- (323.32,76.54) -- (323.32,117.67) -- (288.29,117.67) -- cycle ;
\draw  [color={rgb, 255:red, 65; green, 117; blue, 5 }  ,draw opacity=1 ][dash pattern={on 4.5pt off 4.5pt}] (494.95,76.54) -- (529.98,76.54) -- (529.98,117.67) -- (494.95,117.67) -- cycle ;
\draw  [color={rgb, 255:red, 208; green, 2; blue, 27 }  ,draw opacity=1 ][pattern=_h7e63vsv8,pattern size=6pt,pattern thickness=0.75pt,pattern radius=0pt, pattern color={rgb, 255:red, 208; green, 2; blue, 27}][dash pattern={on 4.5pt off 4.5pt}] (185.13,77.06) -- (253.79,77.06) -- (253.79,117.41) -- (185.13,117.41) -- cycle ;
\draw  [color={rgb, 255:red, 208; green, 2; blue, 27 }  ,draw opacity=1 ][pattern=_tuf4mjvvs,pattern size=6pt,pattern thickness=0.75pt,pattern radius=0pt, pattern color={rgb, 255:red, 208; green, 2; blue, 27}][dash pattern={on 4.5pt off 4.5pt}] (357.01,77.06) -- (459.77,77.06) -- (459.77,116.46) -- (357.01,116.46) -- cycle ;
\draw  [color={rgb, 255:red, 65; green, 117; blue, 5 }  ,draw opacity=1 ][pattern=_ywf33vic5,pattern size=6pt,pattern thickness=0.75pt,pattern radius=0pt, pattern color={rgb, 255:red, 65; green, 117; blue, 5}][dash pattern={on 4.5pt off 4.5pt}] (565.17,77.74) -- (632.33,77.74) -- (632.33,117.14) -- (565.17,117.14) -- cycle ;
\draw  [color={rgb, 255:red, 65; green, 117; blue, 5 }  ,draw opacity=1 ][pattern=_fv2ow3wzu,pattern size=6pt,pattern thickness=0.75pt,pattern radius=0pt, pattern color={rgb, 255:red, 65; green, 117; blue, 5}][dash pattern={on 4.5pt off 4.5pt}] (357.01,77.4) -- (459.77,77.4) -- (459.77,116.8) -- (357.01,116.8) -- cycle ;
\draw [color={rgb, 255:red, 74; green, 144; blue, 226 }  ,draw opacity=1 ]   (60,157.67) -- (60,178.33) ;
\draw [color={rgb, 255:red, 208; green, 2; blue, 27 }  ,draw opacity=1 ]   (101,158.67) -- (101,179.33) ;
\draw [color={rgb, 255:red, 65; green, 117; blue, 5 }  ,draw opacity=1 ]   (141,158.67) -- (141,179.33) ;

\draw (192.67,42.07) node [anchor=north west][inner sep=0.75pt]  [color={rgb, 255:red, 208; green, 2; blue, 27 }  ,opacity=1 ]  {$\mathfrak{I} _{l}( \gamma ''_{1})$};
\draw (384,41.4) node [anchor=north west][inner sep=0.75pt]  [color={rgb, 255:red, 208; green, 2; blue, 27 }  ,opacity=1 ]  {$\mathfrak{I} _{r}( \gamma ''_{1})$};
\draw (384,132.73) node [anchor=north west][inner sep=0.75pt]  [color={rgb, 255:red, 65; green, 117; blue, 5 }  ,opacity=1 ]  {$\mathfrak{I} _{l}( \gamma ''_{2})$};
\draw (567.33,40.73) node [anchor=north west][inner sep=0.75pt]  [color={rgb, 255:red, 65; green, 117; blue, 5 }  ,opacity=1 ]  {$\mathfrak{I} _{r}( \gamma ''_{2})$};
\draw (82,43.07) node [anchor=north west][inner sep=0.75pt]  [color={rgb, 255:red, 74; green, 144; blue, 226 }  ,opacity=1 ]  {$V( \gamma )$};
\draw (282,43.07) node [anchor=north west][inner sep=0.75pt]  [color={rgb, 255:red, 208; green, 2; blue, 27 }  ,opacity=1 ]  {$V( \gamma ''_{1})$};
\draw (488,43.07) node [anchor=north west][inner sep=0.75pt]  [color={rgb, 255:red, 65; green, 117; blue, 5 }  ,opacity=1 ]  {$V( \gamma ''_{1})$};
\draw (65,161.07) node [anchor=north west][inner sep=0.75pt]  [color={rgb, 255:red, 74; green, 144; blue, 226 }  ,opacity=1 ]  {$\gamma $};
\draw (106,157.07) node [anchor=north west][inner sep=0.75pt]  [color={rgb, 255:red, 208; green, 2; blue, 27 }  ,opacity=1 ]  {$\gamma ''_{1}$};
\draw (146,157.07) node [anchor=north west][inner sep=0.75pt]  [color={rgb, 255:red, 65; green, 117; blue, 5 }  ,opacity=1 ]  {$\gamma ''_{2}$};

\end{tikzpicture}
        \caption{Example where $\gamma_1\cup\gamma_2\subset V(\gamma)^c$. Note that $\mathfrak{I}_r(\gamma''_2)$ is infinite.}
    \end{figure}
    
    Given $\gamma''\in\Gamma''_{2,r}$, consider the two connected components of $V(\Gamma''_{2,r})^c\setminus \mathrm{I}_-(\gamma)$ adjacent to $V(\gamma'')$. We denote the one to the left of $V(\gamma'')$ by $\mathfrak{I}_l(\gamma'')$ and the one to the right of $V(\gamma'')$ by $\mathfrak{I}_r(\gamma'')$. By the distancing property, we have that 
    \begin{align*}
        M2^{ar}\leq \min\{\mathrm{diam}(\mathfrak{I}_l(\gamma'')),\mathrm{diam}(\mathfrak{I}_r(\gamma''))\}.
    \end{align*}
    Furthermore, since $\mathrm{I}_-(\gamma)\neq\emptyset$, at least one of them is finite.

    \begin{figure}[hbt!]
        \centering
 
\tikzset{
pattern size/.store in=\mcSize, 
pattern size = 5pt,
pattern thickness/.store in=\mcThickness, 
pattern thickness = 0.3pt,
pattern radius/.store in=\mcRadius, 
pattern radius = 1pt}
\makeatletter
\pgfutil@ifundefined{pgf@pattern@name@_lqqydj2be}{
\pgfdeclarepatternformonly[\mcThickness,\mcSize]{_lqqydj2be}
{\pgfqpoint{0pt}{-\mcThickness}}
{\pgfpoint{\mcSize}{\mcSize}}
{\pgfpoint{\mcSize}{\mcSize}}
{
\pgfsetcolor{\tikz@pattern@color}
\pgfsetlinewidth{\mcThickness}
\pgfpathmoveto{\pgfqpoint{0pt}{\mcSize}}
\pgfpathlineto{\pgfpoint{\mcSize+\mcThickness}{-\mcThickness}}
\pgfusepath{stroke}
}}
\makeatother

 
\tikzset{
pattern size/.store in=\mcSize, 
pattern size = 5pt,
pattern thickness/.store in=\mcThickness, 
pattern thickness = 0.3pt,
pattern radius/.store in=\mcRadius, 
pattern radius = 1pt}
\makeatletter
\pgfutil@ifundefined{pgf@pattern@name@_r272dkp24}{
\pgfdeclarepatternformonly[\mcThickness,\mcSize]{_r272dkp24}
{\pgfqpoint{0pt}{0pt}}
{\pgfpoint{\mcSize+\mcThickness}{\mcSize+\mcThickness}}
{\pgfpoint{\mcSize}{\mcSize}}
{
\pgfsetcolor{\tikz@pattern@color}
\pgfsetlinewidth{\mcThickness}
\pgfpathmoveto{\pgfqpoint{0pt}{0pt}}
\pgfpathlineto{\pgfpoint{\mcSize+\mcThickness}{\mcSize+\mcThickness}}
\pgfusepath{stroke}
}}
\makeatother
\tikzset{every picture/.style={line width=0.75pt}} 

\begin{tikzpicture}[x=0.75pt,y=0.75pt,yscale=-1,xscale=1]

\draw    (23,150) -- (674,150) ;
\draw [color={rgb, 255:red, 74; green, 144; blue, 226 }  ,draw opacity=1 ]   (64,120.25) -- (64,180.25) ;
\draw [color={rgb, 255:red, 74; green, 144; blue, 226 }  ,draw opacity=1 ]   (124,119.75) -- (124,180.25) ;
\draw [color={rgb, 255:red, 74; green, 144; blue, 226 }  ,draw opacity=1 ]   (364,120.2) -- (364,179.8) ;
\draw [color={rgb, 255:red, 74; green, 144; blue, 226 }  ,draw opacity=1 ]   (644,121) -- (644,180.2) ;
\draw  [color={rgb, 255:red, 74; green, 144; blue, 226 }  ,draw opacity=1 ][dash pattern={on 4.5pt off 4.5pt}] (74.67,130.2) -- (114.33,130.2) -- (114.33,169.4) -- (74.67,169.4) -- cycle ;
\draw  [color={rgb, 255:red, 74; green, 144; blue, 226 }  ,draw opacity=1 ][dash pattern={on 4.5pt off 4.5pt}] (374,131) -- (634.43,131) -- (634.43,170.2) -- (374,170.2) -- cycle ;
\draw [color={rgb, 255:red, 208; green, 2; blue, 27 }  ,draw opacity=1 ]   (184,120.6) -- (184,181) ;
\draw [color={rgb, 255:red, 208; green, 2; blue, 27 }  ,draw opacity=1 ]   (244,119.8) -- (244,180.2) ;
\draw  [color={rgb, 255:red, 74; green, 144; blue, 226 }  ,draw opacity=1 ][dash pattern={on 4.5pt off 4.5pt}] (134,130.16) -- (353.07,130.16) -- (353.07,169.6) -- (134,169.6) -- cycle ;
\draw  [color={rgb, 255:red, 208; green, 2; blue, 27 }  ,draw opacity=1 ][pattern=_lqqydj2be,pattern size=6pt,pattern thickness=0.75pt,pattern radius=0pt, pattern color={rgb, 255:red, 208; green, 2; blue, 27}][dash pattern={on 4.5pt off 4.5pt}] (134,130.16) -- (174.2,130.16) -- (174.2,169.8) -- (134,169.8) -- cycle ;
\draw  [color={rgb, 255:red, 208; green, 2; blue, 27 }  ,draw opacity=1 ][dash pattern={on 4.5pt off 4.5pt}] (194,130) -- (233,130) -- (233,169.4) -- (194,169.4) -- cycle ;
\draw  [color={rgb, 255:red, 208; green, 2; blue, 27 }  ,draw opacity=1 ][pattern=_r272dkp24,pattern size=6pt,pattern thickness=0.75pt,pattern radius=0pt, pattern color={rgb, 255:red, 208; green, 2; blue, 27}][dash pattern={on 4.5pt off 4.5pt}] (254.15,130.4) -- (353.8,130.4) -- (353.8,169.8) -- (254.15,169.8) -- cycle ;
\draw [color={rgb, 255:red, 74; green, 144; blue, 226 }  ,draw opacity=1 ]   (574,191.5) -- (574,212.75) ;
\draw [color={rgb, 255:red, 208; green, 2; blue, 27 }  ,draw opacity=1 ]   (612.5,192.5) -- (612.5,213.75) ;

\draw (191,181.4) node [anchor=north west][inner sep=0.75pt]  [color={rgb, 255:red, 208; green, 2; blue, 27 }  ,opacity=1 ]  {$V( \gamma ''_{3})$};
\draw (281.8,181.4) node [anchor=north west][inner sep=0.75pt]  [color={rgb, 255:red, 208; green, 2; blue, 27 }  ,opacity=1 ]  {$\mathfrak{I}_{r}( \gamma ''_{3})$};
\draw (127.8,181.4) node [anchor=north west][inner sep=0.75pt]  [color={rgb, 255:red, 208; green, 2; blue, 27 }  ,opacity=1 ]  {$\mathfrak{I}_{l}( \gamma ''_{3})$};
\draw (76,97.4) node [anchor=north west][inner sep=0.75pt]  [color={rgb, 255:red, 74; green, 144; blue, 226 }  ,opacity=1 ]  {$I_{-}( \gamma )$};
\draw (220,97.4) node [anchor=north west][inner sep=0.75pt]  [color={rgb, 255:red, 74; green, 144; blue, 226 }  ,opacity=1 ]  {$I_{+}( \gamma )$};
\draw (480,97.4) node [anchor=north west][inner sep=0.75pt]  [color={rgb, 255:red, 74; green, 144; blue, 226 }  ,opacity=1 ]  {$I_{-}( \gamma )$};
\draw (575.5,195.4) node [anchor=north west][inner sep=0.75pt]  [color={rgb, 255:red, 74; green, 144; blue, 226 }  ,opacity=1 ]  {$\gamma $};
\draw (614,191.4) node [anchor=north west][inner sep=0.75pt]  [color={rgb, 255:red, 208; green, 2; blue, 27 }  ,opacity=1 ]  {$\gamma ''_{3}$};

\end{tikzpicture}
        \caption{Example where $\gamma''_3\subset \mathrm{I}_+(\gamma)$.}
    \end{figure}
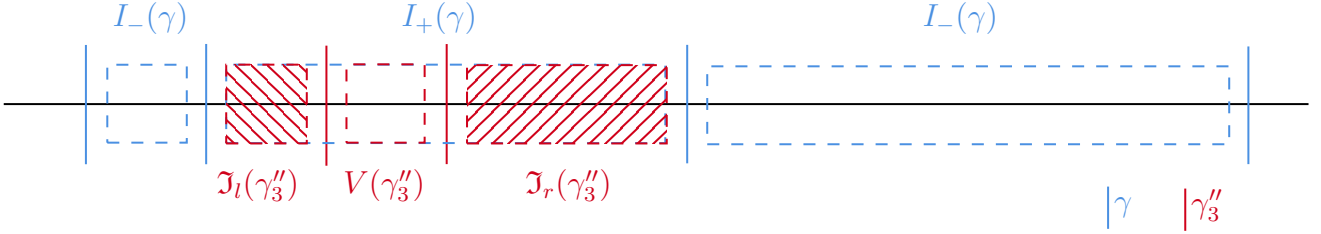
    
    Let $x\in \mathrm{I}_-(\gamma)$. If $x$ lies to the left of $V(\gamma'')$, in which case $\mathfrak{I}_l(\gamma'')$ is finite. Then,
    \begin{align*}
        J_{xy}\leq J_{xz},
    \end{align*}
    for $y\in V(\gamma'')$ and $z\in \mathfrak{I}_l(\gamma'')$. Similarly, if $x$ lies to the right of $V(\gamma'')$, in which case $\mathfrak{I}_r(\gamma'')$ is finite. Then,
    \begin{align*}
        J_{xy}\leq J_{xw}
    \end{align*}
    for $y\in V(\gamma'')$ and $w\in \mathfrak{I}_r(\gamma'')$. 
    
    Recalling that $V(\gamma'')\cap \mathrm{I}_-(\gamma)=\emptyset$, we get that any point $x$ of $\mathrm{I}_-(\gamma)$ lies either to the left or the right of $V(\gamma'')$. Hence,
    \begin{align*}
        \sum_{\substack{x\in \mathrm{I}_-(\gamma)\\y\in V(\gamma'')}}J_{xy}&=\sum_{\substack{x\in \mathrm{I}_-(\gamma)\\y\in V(\gamma'')\\x<V(\gamma'')}}J_{xy}+\sum_{\substack{x\in \mathrm{I}_-(\gamma)\\y\in V(\gamma'')\\x>V(\gamma'')}}J_{xy}\leq \sum_{\substack{x\in \mathrm{I}_-(\gamma)\\y\in V(\gamma'')\\x<V(\gamma'')}}J_{xz}+\sum_{\substack{x\in \mathrm{I}_-(\gamma)\\y\in V(\gamma'')\\x>V(\gamma'')}}J_{xw}\\
        &=\lvert V(\gamma'')\rvert\sum_{\substack{x\in \mathrm{I}_-(\gamma)\\x< V(\gamma'')}}J_{xz}+\lvert V(\gamma'')\rvert\sum_{\substack{x\in \mathrm{I}_-(\gamma)\\x>V(\gamma'')}}J_{xw},
    \end{align*}
    for $z\in \mathfrak{I}_l(\gamma'')$ and $w\in \mathfrak{I}_r(\gamma'')$. If both $\{x\in \mathrm{I}_-(\gamma):x<V(\gamma'')\}$ and $\{x\in \mathrm{I}_-(\gamma):x>V(\gamma'')\}$ are non-empty. Then, by summing over both $z\in \mathfrak{I}_l(\gamma'')$ and $w\in \mathfrak{I}_r(\gamma'')$, we obtain that
    \begin{align*}
        \left(\lvert \mathfrak{I}_l(\gamma'')\rvert+\rvert \mathfrak{I}_r(\gamma'')\rvert\right)\sum_{\substack{x\in \mathrm{I}_-(\gamma)\\y\in V(\gamma'')}}J_{xy}\leq \lvert V(\gamma'')\rvert\sum_{\substack{x\in \mathrm{I}_-(\gamma)\\z\in \mathfrak{I}_l(\gamma'')\sqcup \mathfrak{I}_r(\gamma'')}}J_{xz}.
    \end{align*}
    Finally,
    \begin{align*}
        \sum_{\substack{x\in \mathrm{I}_-(\gamma)\\y\in V(\gamma'')}}J_{xy}\leq \frac{\lvert V(\gamma'')\rvert}{\min\{\lvert \mathfrak{I}_l(\gamma'')\rvert,\lvert \mathfrak{I}_r(\gamma'')\rvert\}}\sum_{\substack{x\in \mathrm{I}_-(\gamma)\\y\in \mathfrak{I}_l(\gamma'')\sqcup \mathfrak{I}_r(\gamma'')}}J_{xy}\leq \frac{2^{r+1}}{M2^{ar}}\sum_{\substack{x\in \mathrm{I}_-(\gamma)\\y\in \mathfrak{I}_l(\gamma'')\sqcup \mathfrak{I}_r(\gamma'')}}J_{xy}.
    \end{align*}
    One can readily check that the same bound is true in case $\mathrm{I}_-(\gamma)$ lies to left, or to the right, of $V(\gamma'')$.

    Noting that each connected component of $V(\Gamma''_{2,r})^c\setminus \mathrm{I}_-(\gamma)$ is adjacent to the volume of at most two contours in $\Gamma''_{2,r}$, we obtain that
    \begin{align*}
        \sum_{\substack{x\in \mathrm{I}_-(\gamma)\\y\in V(\Gamma''_{2,r})}}J_{xy}&=\sum_{\gamma''\in\Gamma''_{2,r}}\sum_{\substack{x\in \mathrm{I}_-(\gamma)\\y\in V(\gamma'')}}J_{xy}\leq \sum_{\gamma''\in\Gamma''_{2,r}}\frac{2^{r+1}}{M2^{ar}}\sum_{\substack{x\in \mathrm{I}_-(\gamma)\\y\in \mathfrak{I}_l(\gamma'')\sqcup \mathfrak{I}_r(\gamma'')}}J_{xy}\\
        &= \frac{2}{M2^{(a-1)r}}\sum_{\gamma''\in\Gamma''_{2,r}}\sum_{\substack{x\in \mathrm{I}_-(\gamma)\\y\in \mathfrak{I}_l(\gamma'')\sqcup \mathfrak{I}_r(\gamma'')}}J_{xy}\leq \frac{4}{M2^{(a-1)r}}\sum_{\substack{x\in \mathrm{I}_-(\gamma)\\y\in V(\Gamma''_{2,r})^c\setminus \mathrm{I}_-(\gamma)}}J_{xy}\\
        &\leq \frac{2}{M2^{(a-1)r}}H_J(\gamma).
    \end{align*}
    Summing over $r$, we get
    \begin{align}
        \sum_{\substack{x\in \mathrm{I}_-(\gamma)\\ y\in V(\Gamma''_2)}}J_{xy}\leq \sum_{r=0}^{n_0(\gamma)}\sum_{\substack{x\in \mathrm{I}_-(\gamma)\\y\in V(\Gamma''_{2,r})}}J_{xy}\leq\frac{2}{M}\left[\sum_{r=0}^\infty\frac{1}{2^{(a-1)r}}\right]H_J(\gamma).
    \end{align}
    Proceeding in the exact same way for $\Gamma'$, we obtain the bound
    \begin{align}
        \sum_{\substack{x\in V(\Gamma')\\y\in \mathrm{I}_-(\gamma)^c}}J_{xy}\leq\frac{1}{M}\frac{2^a}{2^{a-1}-1}H_J(\gamma).
    \end{align}
    Finally,
    \begin{align*}
        H_J(\Gamma)-H_J(\Gamma\hspace{-0.1cm}\setminus\hspace{-0.1cm}\gamma)&\geq H_J(\gamma)-\sum_{\substack{x\in V(\Gamma')\\y\in \mathrm{I}_-(\gamma)^c}}4J_{xy}-\sum_{\substack{x\in \mathrm{I}_-(\gamma)\\y\in V(\Gamma'')}}4J_{xy}\\
        &\geq H_J(\gamma)-\sum_{\substack{x\in V(\Gamma')\\y\in \mathrm{I}_-(\gamma)^c}}4J_{xy}-\sum_{\substack{x\in \mathrm{I}_-(\gamma)\\y\in V(\Gamma''_1)}}4J_{xy}-\sum_{\substack{x\in \mathrm{I}_-(\gamma)\\y\in V(\Gamma''_2)}}4J_{xy}\\
        &\geq H_J(\gamma)-\frac{4}{M}\frac{2^a}{2^{a-1}-1}H_J(\gamma)-4\alpha M^{1-\alpha}H_J(\gamma)-\frac{4}{M}\frac{2^a}{2^{a-1}-1}H_J(\gamma)\\
        &=\left(1-\frac{1}{M}\frac{2^{a+3}}{2^{a-1}-1}-4\alpha M^{1-\alpha}\right)H_J(\gamma).
    \end{align*}
    The lemma has, thus, been proved with $C(a)=2^{a+3}\left(2^{a-1}-1\right)^{-1}$.
    \end{proof}

\subsection{A geometric way of estimating \texorpdfstring{$H_J(\gamma)$}{TEXT}}\label{geometricway}

In this subsection and the next one, we follow in detail the arguments developed in \cite{Frohlich.Spencer.82} introducing a second type of cover size dependent on the distancing parameters $(M,a)$. We establish its relation with the Hamiltonian here, and obtain a criterion for models where phase transition may be proved by the energy-entropy argument.

Given $\gamma\in\Gamma$ an exterior contour, we define
\begin{align*}
    \mathcal{I}'_n(\gamma)\coloneqq\{I\in\mathcal{I}_n(\gamma):\mathrm{dist}(I,I')\geq 2M2^{an}, \forall I'\in\mathcal{I}_n(\gamma), I'\neq I\},
\end{align*}
where $1\leq n\leq n_0(\gamma)\coloneqq\lfloor\log_2\mathrm{diam}(\gamma)\rfloor$ and $\mathcal{I}_n(\gamma)$ denotes the canonical open covering of $\gamma$ by open intervals with diameter $2^n$ introduced in the previous section. We call $I\in\mathcal{I}'_n(\gamma)$ an \textit{isolated interval}. Finally, we arrive at the notion of \textit{isolated covering size}:
\begin{align}
    \mathcal{N}'(\gamma)\coloneqq\lvert\gamma\rvert+\sum_{n=1}^{n_0(\gamma)}\lvert\mathcal{I}'_n(\gamma)\rvert.
\end{align}
Our objective in this subsection is to show that $H_J(\gamma)\geq \epsilon_J\mathcal{N}'(\gamma)$, where $\epsilon_J>0$ is a constant depending only on the interaction. That is, independent of both $1<a<2$ and $M>1$. In which case, the notion of distant covering size $\mathcal{N}'(\gamma)$ carries all the contour-geometric information. It is at this point that the argument for the existence of phase transition may fail, in case the coupling $J$ decays too fast.

\begin{remark}
    We note that $\gamma\cap I$ is odd for $I\in\mathcal{I}'_n(\gamma)$, $1\leq n\leq n_0(\gamma)$. Indeed, consider the decomposition of $\gamma$ given by $\{\gamma\cap I,\gamma\setminus I\}$, then
    \begin{align*}
        \mathrm{dist}(\gamma\cap I,\gamma\setminus I)&\geq\mathrm{dist}(I,\gamma\setminus I)\geq\min_{\substack{I'\in \mathcal{I}_n(\gamma)\\I'\neq I}}\mathrm{dist}(I,I')\geq 2M2^{an}\\
        &>M2^{an}=M\mathrm{diam}(I)^a\geq M(\min\{\mathrm{diam}(\gamma\cap I),\mathrm{diam}(\gamma\setminus I)\})^a.
    \end{align*}
    Since $\gamma$ is irreducible, both $\gamma\cap I$ and $\gamma\setminus I$ have odd cardinality.
\end{remark}

To produce the constant $\epsilon_J$, it is convenient to introduce the following sequence:
\begin{align}
    D_n(J)\coloneqq\sum_{\substack{x\in I_n^-\\y\in I_n^+}}J_{xy},
\end{align}
where $n\geq 1$. $I_n^+=[2^{n-1},2^n-1]$ and $I_n^-=-I_n^+$. It turns out a bound of the type $H_J(\gamma)\geq\epsilon_J \mathcal{N}'(\gamma)$, with $\epsilon_J>0$, can be obtained provided 
\begin{align}
    \inf_{n\in\mathbb{N}}D_n(J)\coloneqq D(J)>0.
\end{align}
In fact, as we shall show, one obtains $\epsilon_J=2\min\{J(1),D(J)\}$. We recall that $J(1)=1$.

\begin{remark}
    Since we suppose that $r\mapsto J(r)$ is decreasing, we have that
    \begin{align*}
        2^{2(n-1)}J(2^{n+1})\leq D_n(J)\leq 2^{2(n-1)}J(2^n).
    \end{align*}
    In which case, $D(J)>0$ if, and only if, 
    \begin{align}
        \liminf_{r\rightarrow\infty}r^2J(r)>0.
    \end{align}
    This is certainly true if $J(r)=r^{-\alpha}$ with $1<\alpha\leq 2$.
\end{remark}

Let $1\leq n\leq n_0(\gamma)$ and $(a,b)=I\in\mathcal{I}'_n(\gamma)$. We define $\mathrm{I}_+=[b,b+2^{n-1}-1]$ and $\mathrm{I}_-=[a-2^{n-1}+1,a]$. Note that $\bar{I}\coloneqq I_-\sqcup I\sqcup I_+$ is a closed interval with diameter $2^{n+1}-2$. Furthermore, $I$ and $\bar{I}$ are centered around the same point. By translation invariance, it holds that  
\begin{align}\label{krak}
    \sum_{\substack{x\in \mathrm{I}_-\\y\in \mathrm{I}_+}}J_{xy}=D_n(J).
\end{align}
Note that $\gamma\cap \mathrm{I}_+=\gamma\cap \mathrm{I}_-=\emptyset$. To prove the estimate $H_J(\gamma)\geq\epsilon_J\mathcal{N}'(\gamma)$, we need to make sure that, by summing \eqref{krak} over all isolated intervals, no pair $(x,y)$ appears twice.

\begin{lemma}\label{intervalospnc}
    Let $I\in\mathcal{I}'_n(\gamma)$ and $I'\in\mathcal{I}'_m(\gamma)$, with $1\leq n,m\leq n_0(\gamma)$ and $I\neq I'$. Then,
    \begin{align}\label{intervals}
        [(I_-\times I_+)\cup (I_+\times I_-)]\cap[(I_-'\times I_+')\cup (I_+'\times I_-')]=\emptyset.
    \end{align}
\end{lemma}

\begin{proof}
If $n=m$, the proof is direct since $\mathrm{dist}(I,I')\geq 2M2^{an}$. Therefore, suppose without loss of generality that $n<m$. If $I\cap I'=\emptyset$, there exists $I''\in\mathcal{I}_m(\gamma)$ such that $I\cap I''\neq\emptyset$. By definition, it holds that 
\begin{align*}
    \mathrm{dist}(I,I')\geq\mathrm{dist}(I',I'')-\mathrm{diam}(I)\geq 2M2^{am}-2^n> 2^m.
\end{align*}
In which case, $I_-\cap I'_-=I_+\cap I'_-=I_+\cap I'_-=I_+\cap I'_+=\emptyset$.

In general, let $I$ and $I'$ be real intervals such that $\mathrm{diam}(I)<\mathrm{diam}(I')$. Let $\bar{I}$ be a third interval such that $\bar{I}=I_-\sqcup I\sqcup I_+$, where $I_-,I_+$ are non-empty intervals satisfying $\mathrm{dist}(I_-,I_+)>0$. If $\mathrm{diam}(\bar{I})<\mathrm{diam}(I')$ and $I\cap I'\neq\emptyset$, then there exists $i\in\{-,+\}$ such that $I_i\subset I'$. In our case, $\mathrm{diam}(\bar{I})=2^{n+1}-2$ and $\mathrm{diam}(I')=2^m\geq 2^{n+1}$. Therefore, if $I\cap I'\neq \emptyset$, either $I_-\subset I'$, in which case, $I'_+\cap I_-=I'_-\cap I_-=\emptyset$, or $I_+\subset I'$, in which case, $I'_+\cap I_+=I'_-\cap I_+=\emptyset$. The lemma has, thus, been proved.

\begin{figure}[hbt!]
        \centering
        \input{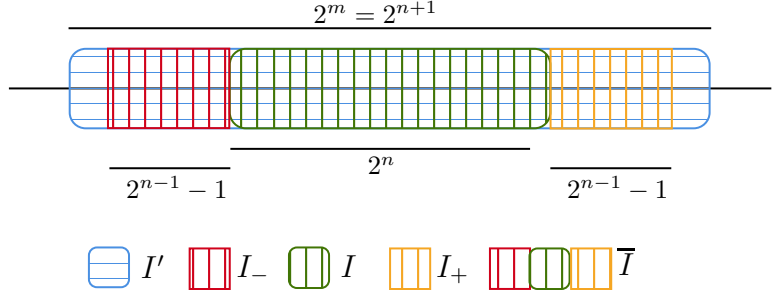}
        \caption{Limiting case where $m=n+1$ and, $I$ and $I'$ are concentric.}
    \end{figure}
\end{proof}

Since $\gamma$ is irreducible and $I$ is isolated, it is true that $\lvert\gamma\cap I\rvert$ is odd. This implies that $\chi_\gamma(x,y)=1$ if $x\in \mathrm{I}_+$ and $y\in \mathrm{I}_-$.  By Lemma \ref{intervalospnc}, we obtain that
\begin{align}\label{disjointunion}
    \bigsqcup_{n=1}^{n_0(\gamma)}\bigsqcup_{I\in\mathcal{I}'_n(\gamma)}\mathrm{I}_-\times \mathrm{I}_+\subset \{(x,y):\chi_\gamma(x,y)=1\},
\end{align}
where the squared union symbol denotes that it is an union of disjoint sets. We can now prove the bound $H_J(\gamma)\geq\epsilon_J\mathcal{N}'(\gamma)$, for $\epsilon_J>0$, provided $\liminf_{r\rightarrow\infty}r^2J(r)>0.$ 

\begin{lemma}
    Suppose that $D(J)>0$. Then, there exists $\epsilon_J>0$, independent of both $1<a<2$ and $M>1$, such that
    \begin{align}
        H_J(\gamma)\geq \epsilon_J\mathcal{N}'(\gamma),
    \end{align}
    for any irreducible contour $\gamma$.
\end{lemma}

\begin{proof}
    By Equation \eqref{disjointunion}, we have that
    \begin{align*}
        H_J(\gamma)&=2\sum_{x<y}J_{xy}\chi_\gamma(x,y)=2J(1)\lvert\gamma\rvert+2\sum_{x+1<y}J_{xy}\chi_{\gamma}(x,y)\\
        &\geq 2\lvert\gamma\rvert+2\sum_{n=1}^{n_0(\gamma)}\sum_{I\in\mathcal{I}'_n(\gamma)}\sum_{\substack{x\in \mathrm{I}_-\\y\in \mathrm{I}_+}}J_{xy}\geq 2\lvert\gamma\rvert + 2D(J)\sum_{n=1}^{n_0(\gamma)}\sum_{I\in\mathcal{I}'_n(\gamma)}1\\
        &\geq 2\min\{1,D(J)\}\mathcal{N}'(\gamma).
    \end{align*}
    The lemma has, thus, been proved.
\end{proof}

\subsection{Entropy estimates. The connection between cover sizes}\label{entropysizes}

All that remains is establishing a quantitative relation between $\mathcal{N}(\gamma)$ and $\mathcal{N}'(\gamma)$, an estimate of the type $\mathcal{N}(\gamma)\leq c(M,a)\mathcal{N}'(\gamma)$. We define $s:\mathbb{Z}\rightarrow\mathbb{Z}$
\begin{align}
    s(n)=\left\lfloor\frac{n-\log_2(8M)}{a}\right\rfloor,
\end{align}
which is a natural modification to the function connecting scales present in \cite{Frohlich.Spencer.82}. A similar idea is used in the multidimensional case, see \cite{Affonso.2021} and \cite{Johanes}. This function is chosen so that the following recursive relation is true for $n\geq a+\log_2(8M)$:
\begin{align*}
    \lvert \mathcal{I}_n(\gamma)\rvert\leq \frac{1}{2}\lvert \mathcal{I}_{s(n)}\rvert+\frac{1}{2}\lvert\mathcal{I}'_{s(n)}\rvert.
\end{align*}
This motivates us to define
\begin{align}\label{definicaonbarra}
    \bar{n}=\lceil a+\log_2(8M)\rceil,
\end{align}
so that $s(n)>0$ if, and only if, $n\geq \bar{n}$. Employing Claim 4.2 and the same reasoning as in Lemma 4.1 in \cite{https://doi.org/10.48550/arxiv.1711.04720}, we obtain the following recursive relation:

\begin{lemma}
    If $n\geq\bar{n}$, it holds that
    \begin{align}\label{connectcoversizes}
        \left\lvert\mathcal{I}_n(\gamma)\right\rvert\leq \frac{1}{2}\left\lvert\mathcal{I}_{s(n)}(\gamma)\right\rvert+\frac{1}{2}\left\lvert\mathcal{I}'_{s(n)}(\gamma)\right\rvert.
    \end{align}
\end{lemma}

\begin{remark}
    For a further generalization, which is very useful in the multidimensional case, see Proposition 3.14 in \cite{Affonso.2021}.
\end{remark}

\begin{lemma}
    Let $M$ be sufficiently large and $1<a<2$. Then,
    \begin{align}\label{formulafeiabagarai}
        \mathcal{N}(\gamma)\leq c(M,a)\mathcal{N}'(\gamma),
    \end{align}
    for some $c(M,a)>0$.
\end{lemma}

\begin{proof}
Define $l:\{n\in\mathbb{N}:n\geq\bar{n}\}\rightarrow\mathbb{N}$ by
\begin{align}
    l(n)=\max\{m\in\mathbb{N}:s^m(n)>0\},
\end{align}
where $s^m$ denotes $m$ compositions of $s$ with itself. By iterating \eqref{connectcoversizes}, we get that 
\begin{align*}
    \lvert\mathcal{I}_n(\gamma)\rvert\leq 2^{-l(n)}\left\lvert\mathcal{I}_{s^{l(n)}(n)}(\gamma)\right\rvert+\sum_{m=1}^{l(n)}2^{-m}\left\lvert\mathcal{I}'_{s^m(n)}(\gamma)\right\rvert.
\end{align*}
Hence,
\begin{align*}
    \mathcal{N}(\gamma)&=\sum_{n=0}^{n_0(\gamma)}\lvert\mathcal{I}_n(\gamma)\rvert\leq \bar{n}\lvert\gamma\rvert+\sum_{n=\bar{n}}^{n_0(\gamma)}\left[2^{-l(n)}\lvert\gamma\rvert+\sum_{m=1}^{l(n)}2^{-m}\left\lvert\mathcal{I}'_{s^m(n)}(\gamma)\right\rvert\right]\\
    &\leq\left[\bar{n}+\sum_{n=\bar{n}}^{n_0(\gamma)}2^{-l(n)}\right]\mathcal{N}'(\gamma)+\sum_{n=\bar{n}}^{n_0(\gamma)}\sum_{m=1}^{l(n)}2^{-m}\left\lvert\mathcal{I}'_{s^m(n)}(\gamma)\right\rvert\\
\end{align*}
Good estimates of $l(n)$ are, therefore, needed.

To do so, we begin by estimating $s(n)$. Let $n\geq\bar{n}$, we define
\begin{align}
    \underline{s}(n)\coloneqq\frac{n-\log_2(8M)}{a}-1\leq s(n)\leq \frac{n-\log_2(8M)}{a}\eqqcolon \overline{s}(n).
\end{align}
Similarly, we define
\begin{align}
    \underline{l}(n)\coloneqq\max\{m\in\mathbb{N}: \underline{s}^m(n)\geq 1\}\leq l(n)\leq \max\{m\in\mathbb{N}:\overline{s}^m(n)\geq 1\}\eqqcolon \overline{l}(n).
\end{align}
A direct linear recurrence argument yields
\begin{align}
    &\overline{s}^m(n)=a^{-m}n-\left[\frac{1-a^{-m}}{a-1}\right]\log_2(8M)\label{estimatings}\\
    &\underline{s}^m(n)=a^{-m}n-\left[\frac{1-a^{-m}}{a-1}\right]\left[\log_2(8M)+a\right]\label{estimatings2},
\end{align}
for $m\geq 0$. Hence,
\begin{align}\label{boundln}
    -1+\log_a\frac{(a-1)n+a+\log_2(8M)}{(2a-1)+\log_2(8M)}\leq l(n)\leq \log_a\frac{(a-1)n+\log_2(8M)}{(a-1)+\log_2(8M)}.
\end{align}
It is also necessary to estimate the cardinality of $S_{m,j}\coloneqq\{n:s^m(n)=j\}$, for $m,j\geq 1$. By \eqref{estimatings} and \eqref{estimatings2}, we arrive at
\begin{align*}
    \bar{s}^m(n)-\underline{s}^m(n)=a\frac{1-a^{-m}}{a-1}\leq \frac{a}{a-1}.
\end{align*}
Let $n\in S_{m,j}$, then
\begin{align}
    j=s^m(n)\leq \bar{s}^m(n)\leq \underline{s}^m(n)+\frac{a}{a-1}\leq s^m(n)+\frac{a}{a-1}=j+\frac{a}{a-1}.
\end{align}
The first and third inequalities imply that
\begin{align*}
    n\geq a^mj-\frac{a}{a-1}+\frac{a^m-1}{a-1}\log_2(8M).
\end{align*}
The second and fourth inequalities imply that
\begin{align*}
    n\leq a^mj+\frac{a^{m+1}}{a-1}+\frac{a^m-1}{a-1}\log_2(8M).
\end{align*}
Hence,
\begin{align}\label{bounds}
    \lvert S_{m,j}\rvert\leq \frac{3a}{a-1}a^m.
\end{align}

On one hand, by \eqref{boundln}, we obtain that
\begin{align*}
    \sum_{n=\bar{n}}^{n_0(\gamma)}2^{-l(n)}\leq \sum_{n=\bar{n}}^\infty 2^{-l(n)}&\leq 2\sum_{n=\bar{n}}^\infty\left[\frac{(2a-1)+\log_2(8M)}{(a-1)n+a+\log_2(8M)}\right]^{\log_a2}\\
    &\leq 2\left[\frac{2a-1}{a-1}\right]^{\log_a2}\left[\log_2(8M)\right]^{\log_a2}\sum_{n=1}^\infty\frac{1}{n^{\log_a2}}\\
    &\leq 2\left[\frac{2a}{a-1}\right]^{\log_a2}\zeta(\log_a2)\left[\log_2(8M)\right]^{\log_a2}.
\end{align*}
On the other hand, by \eqref{bounds}, we obtain that
\begin{align*}
    \sum_{n=\bar{n}}^{n_0(\gamma)}\sum_{m=1}^{l(n)}2^{-m}\left\lvert\mathcal{I}'_{s^m(n)}(\gamma)\right\rvert&\leq \sum_{n=\bar{n}}^{n_0(\gamma)}\sum_{m=1}^\infty\sum_{j=1}^{n_0(\gamma)}2^{-m}\left\lvert\mathcal{I}'_j(\gamma)\right\rvert\chi_{\{s^m(n)=j\}}\leq\sum_{j=1}^{n_0(\gamma)}\left[\sum_{m=1}^\infty2^{-m}\left\lvert S_{m,j}\right\rvert\right]\left\lvert\mathcal{I}'_j(\gamma)\right\rvert\\
    &\leq \sum_{j=1}^{n_0(\gamma)}\left[\frac{3a}{a-1}\sum_{m=1}^\infty\left(\frac{a}{2}\right)^m\right]\left\lvert\mathcal{I}'_j(\gamma)\right\rvert\leq\frac{6a}{(a-1)(2-a)}\mathcal{N}'(\gamma).
\end{align*}

Recalling \eqref{definicaonbarra} and that $1<a<2$, we have that
\begin{align*}
    \bar{n}\leq 1+a+\log_2(8M)\leq 2\left[\frac{2a}{a-1}\right]^{\log_a2}\zeta\left(\log_a2\right)[\log_2(8M)]^{\log_a2}.
\end{align*}
Taking $M>1$ sufficiently large, so that 
\begin{align*}
    \frac{3}{2-a}\leq \log_2(8M),
\end{align*}
we bound the final term
\begin{align*}
    \frac{6a}{(a-1)(2-a)}\leq 2\left[\frac{2a}{a-1}\right]^{\log_a2}\zeta\left(\log_a2\right)[\log_2(8M)]^{\log_a2}.
\end{align*}
Finally, we arrive at \eqref{formulafeiabagarai} with
\begin{align*}
    c(M,a)=6\left[\frac{2a}{a-1}\right]^{\log_a2}\zeta\left(\log_a2\right)[\log_2(8M)]^{\log_a2}.
\end{align*}
\end{proof}

\subsection{The Peierls' argument proving phase transition}\label{peierls}

Let $1<a<2$ and $1<\alpha\leq 2$. One could take $a=\frac{3}{2}$ as in \cite{Frohlich.Spencer.82}. Choose $M=M(\alpha,a)$ such that 
\begin{align*}
    H_J(\Gamma)-H_J(\Gamma\hspace{-0.1cm}\setminus\hspace{-0.1cm}\gamma)\geq \frac{1}{2}H_J(\gamma),
\end{align*}
where $\Gamma\coloneqq\Gamma(\sigma,M,a)$ and $\gamma\in\Gamma$ is any external $(M,a)$-contour. Then,
\begin{align}
    \mathcal{N}(\gamma)\leq 6\left[\frac{2a}{a-1}\right]^{\log_a2}\zeta\left(\log_a2\right)(\log_2(8M))^{\log_a2}\frac{1}{\epsilon_J}H_J(\gamma)\coloneqq K(\alpha,a,M)H_J(\gamma)=K H_J(\gamma).
\end{align}
By standard arguments described in Section \ref{outline}, we prove the following theorem:


\begin{theorem}
    Let $1<\alpha\leq 2$. There exists $\bar{\beta}=\bar{\beta}(\alpha)$ such that
    \begin{align*}
        \langle\sigma_0\rangle_{\alpha,\beta}^-\neq\langle\sigma_0\rangle^+_{\alpha,\beta},
    \end{align*}
    if $\beta\geq\bar{\beta}$. That is, there is phase transition.
\end{theorem}

\section{Stability of phase transition under small perturbations}\label{stability}

Let $h:\mathbb{Z}\rightarrow\mathbb{R}_{\geq0}$ be a function describing an external field. We define $E_h:\Omega^*\rightarrow\mathbb{R}_{\geq 0}$ by
\begin{align}
    E_h(\gamma)=\sum_{x\in \mathrm{I}_-(\gamma)}2h_x.
\end{align}
Similarly, if $\Gamma\Subset\Omega^*$ with $\gamma\cap\gamma'=\emptyset$ for any distinct $\gamma,\gamma'\in\Gamma$. We define
\begin{align}
    E_h(\Gamma)\doteq \sum_{x\in N(\Gamma)}2h_x=E_h(\cup_{\gamma\in\Gamma}\gamma).
\end{align}

Once again, we introduce slight modifications in the way of writing the field contribution to remove the necessity of introducing boxes. We shall write it directly via collections of spin flips
\begin{align}
    &H^+_{J,h}(\gamma)=2\hspace{-0.3 cm}\sum_{\substack{x\in \mathrm{I}_-(\gamma)\\y\in \mathrm{I}_-(\gamma)^c}}J_{xy}+2\hspace{-0.3 cm}\sum_{x\in \mathrm{I}_-(\gamma)}h_x=H_J(\gamma)+E_h(\gamma),\\[0.3 cm]
    &H^-_{J,h}(\gamma)=H_J(\gamma)-E_h(\gamma).
\end{align}

\begin{remark}
    Since Gibbs measures are not changed by adding constants, these two formulae are equivalent to writing Hamiltonians in the following way
    \begin{align*}
        &H^+_{J,h}(\sigma)=\sum_{x<y}J_{xy}(1-\sigma_x\sigma_y)+\sum_xh_x(1-\sigma_x),\\
        &H^-_{J,h}(\sigma)=\sum_{x<y}J_{xy}(1-\sigma_x\sigma_y)-\sum_xh_x(1+\sigma_x).
    \end{align*}
    Note that there is no need for $h$ to be non-negative.
\end{remark}

For $A\Subset\mathbb{Z}$, define $\widehat{h}:\mathbb{Z}\rightarrow\mathbb{R}_{\geq 0}$ by
\begin{align}
    \widehat{h}_x\doteq\begin{cases}
        h_x &\textrm{ if }x\in A^c,\\
        0 &\textrm{ otherwise.}
    \end{cases}
\end{align}
We will not carry $A$ in the notation of the truncated field since for our case of interested, it is sufficient to take $A=[-R,R]$ with $R>0$ suitably chosen.

Let us introduce a criterion that is crucial for the stability of the phase diagram. We shall suppose that there exist $\eta\coloneqq \eta(J,\widehat{h})\in (0,1)$ and $A\Subset\mathbb{Z}$ such that
\begin{align}\label{criterion1}
    E_{\hat{h}}(\gamma)\leq \eta H_J(\gamma),
\end{align}
for all $\gamma\in\Omega^*$.

\begin{remark}
    Since $\sum_x\lvert h_x - \widehat{h}_x \rvert<\infty$, there is phase transition for $(J,h)$ if, and only if, there is phase transition for $(J,\widehat{h})$. For proof that finite energy perturbations don't destroy phase transition, see Georgii's book \cite{georgii.gibbs.measures}.
\end{remark}

\begin{lemma}
    Let $\Gamma\Subset\Omega^*$ be well-ordered and $\gamma\in\mathcal{E}_{\mathrm{ext}}(\Gamma)$. Then,
    \begin{align}
        \lvert E_h(\Gamma)-E_h(\Gamma\hspace{-0.1cm}\setminus\hspace{-0.1cm}\gamma)\rvert\leq E_h(\gamma).
    \end{align}
\end{lemma}

\begin{proof}
    Recalling the properties of well-ordered collections of collections of spin flips, we obtain
    \begin{align*}
        E_h(\Gamma)-E_h(\Gamma\hspace{-0.1cm}\setminus\hspace{-0.1cm}\gamma)&=\sum_{x\in N(\Gamma)}2h_x-\hspace{-0.3cm}\sum_{x\in N(\Gamma\setminus\gamma)}\hspace{-0.3 cm}2h_x \quad \text{(Lemma \ref{posneg})}\\[0.2 cm]
        &=\sum_{x\in N(\Gamma(\gamma))}\hspace{-0.3cm}2h_x-\hspace{-0.3cm}\sum_{x\in N(\Gamma_+(\gamma))}\hspace{-0.3cm}2h_x-\hspace{-0.3cm}\sum_{x\in N(\Gamma_-(\gamma))}\hspace{-0.3cm}2h_x \quad \text{(Lemma \ref{posneg})}\\[0.2 cm]
        &=\hspace{-0.7cm}\sum_{x\in \mathrm{I}_-(\gamma)\setminus N(\Gamma_-(\gamma))}\hspace{-0.6cm}2h_x-\sum_{x\in N(\Gamma_-(\gamma
        ))}\hspace{-0.3cm}2h_x,
    \end{align*}
    where we recall that $\Gamma_\omega(\gamma)=\{\gamma'\in\Gamma:\gamma'\subset \mathrm{I}_\omega(\gamma)\}$, $\omega\in\{-,+\}$, and $\Gamma(\gamma)=\{\gamma\}\sqcup\Gamma_+(\gamma)\sqcup\Gamma_-(\gamma)$.

    Then, by the non-negativity of $h$, we obtain that
    \begin{align*}
        \lvert E_h(\Gamma)-E_h(\Gamma\hspace{-0.1cm}\setminus\hspace{-0.1cm}\gamma)\rvert\leq \sum_{x\in \mathrm{I}_-(\gamma)\setminus N(\Gamma_-(\gamma))}\hspace{-0.7cm}2h_x+\sum_{x\in N(\Gamma_-(\gamma))}\hspace{-0.3cm}2h_x=\sum_{x\in \mathrm{I}_-(\gamma)}\hspace{-0.1 cm}2h_x=E_h(\gamma).
    \end{align*}
    The lemma has, thus, been proved.
\end{proof}

\begin{remark}
    This lemma holds for $\hat{h}$ for any $A\Subset\mathbb{Z}$.
\end{remark}

\begin{theorem}
    Let $h:\mathbb{Z}\rightarrow\mathbb{R}_{\geq 0}$ be a function and $\alpha\in (1,2]$. Suppose there exist $A\Subset\mathbb{Z}$ and $\eta\in (0,1)$ such that
    \begin{align*}
        \eta\sum_{\substack{x\in \Lambda\\y\in \Lambda^c}}\lvert x-y\rvert^{-\alpha}\geq \sum_{x\in \Lambda}\hat{h}_x,
    \end{align*}
    for all $\Lambda\Subset\mathbb{Z}$. Then, there exists $\bar{\beta}=\bar{\beta}(J,\hat{h})$ such that
    \begin{align*}
        \langle\sigma_0\rangle^-_{\alpha,\hat{h},\beta}\neq\langle\sigma_0\rangle^+_{\alpha,\hat{h},\beta},
    \end{align*}
    if $\beta\geq\bar{\beta}$. That is, the introduction of a small enough decaying field does not destroy the phenomenon of phase transition.
\end{theorem}

\begin{proof}
    Let us show that if such an uniform bound is valid, the measures with homogeneous boundary condition are different for inverse temperature $\beta$ sufficiently large. Once again we shall be working with $(M,a)$-partitions. Hence, let $\Gamma=\Gamma(\sigma,M,a)$ denote an $(M,a)$-partition of $\partial\sigma$, where $\sigma\in\Omega^+$. Let $\gamma\in\mathcal{E}_{\mathrm{ext}}(\Gamma)$. Then, by \eqref{energyestimates} and \eqref{criterion1},
    \begin{align*}
        H^+_{J,\hat{h}}(\Gamma)-H^+_{J,\hat{h}}(\Gamma\hspace{-0.1cm}\setminus\hspace{-0.1cm}\gamma)&\geq\left[H_J(\Gamma)-H_J(\Gamma\hspace{-0.1cm}\setminus\hspace{-0.1cm}\gamma)\right]+[E_{\hat{h}}(\Gamma)-E_{\hat{h}}(\Gamma\hspace{-0.1cm}\setminus\hspace{-0.1cm}\gamma)]\\
        &\geq H_J(\gamma)\left(1-\frac{C}{M}-4\alpha M^{1-\alpha}\right)-E_{\hat{h}}(\gamma)\\
        &\geq H_J(\gamma)\left(1-\eta-\frac{C}{M}-4\alpha M^{1-\alpha}\right).
    \end{align*}
    In the case of configurations with $-$-boundary condition, we obtain the exact same bound:
    \begin{align}
        H^-_{J,\hat{h}}(\Gamma)-H^-_{J,\hat{h}}(\Gamma\hspace{-0.1cm}\setminus\hspace{-0.1cm}\gamma)\geq H_J(\gamma)\left(1-\eta-\frac{C}{M}-4\alpha M^{1-\alpha}\right).
    \end{align}

    Choose $M=M(\alpha,a)$ such that
    \begin{align*}
        1-\eta-\frac{C}{M}-4\alpha M^{1-\alpha}\geq \frac{1-\eta}{2}.
    \end{align*}
    Proceeding in the exact same way as in Subsection \ref{outline}, we arrive at
    \begin{align*}
        \frac{1}{2}\langle 1+\sigma_0\rangle^-_{\alpha,\hat{h},\beta,L}\leq \sum_{\substack{\gamma_0\in\Omega_L^*\\\gamma_0\textrm{ irreducible}\\0\in V(\gamma_0)}}e^{-\frac{1-\eta}{2}\beta H_J(\gamma_0)}\leq \sum_{R=1}^\infty e^{\left(C_2-\frac{1-\eta}{2K}\beta\right)R}=\frac{e^{\left(C_2-\frac{1-\eta}{2K}\beta\right)}}{1-e^{\left(C_2-\frac{1-\eta}{2K}\beta\right)}}\xrightarrow{\beta\rightarrow\infty}0.
    \end{align*}
    Similarly, we have that
    \begin{align*}
        \frac{1}{2}\langle 1-\sigma_0\rangle^+_{\alpha,\hat{h},\beta,L}\leq\frac{e^{\left(C_2-\frac{1-\eta}{2K}\beta\right)}}{1-e^{\left(C_2-\frac{1-\eta}{2K}\beta\right)}}\xrightarrow{\beta\rightarrow\infty}0.
    \end{align*}
    \end{proof}

\subsection{A simple min-max problem}

Given a coupling $J$ and a non-negative external field $h$, it is rather straightforward, in the one-dimensional case, to check whether the introduction of the external field destroys, or not, the phase diagram. Let us provide the reason why that is the case.

Suppose $\gamma\in\Omega^*$ is such that $\lvert \mathrm{I}_-(\gamma)\rvert=n$. What sort of a lower bound can obtained for $H_J(\gamma)$? Since the interaction is ferromagnetic and $r\mapsto J(r)$ is non-increasing, it must be the energy of a configuration describing a single bubble of minuses of size $n$ (we assume $+$-boundary conditions). In the one-dimensional case, there is a single way of assembling such a bubble given the very simple geometry of $\mathbb{Z}$. We conclude that
\begin{align*}
    \inf_{\substack{\gamma\in\Omega^*\\ \lvert \mathrm{I}_-(\gamma)\rvert=n}}H_J(\gamma)=H_J\left(\left\{\frac{1}{2},n+\frac{1}{2}\right\}\right)\eqqcolon \underline{H}_J(n).
\end{align*}
On the other hand, as for the external field, one obtains the following
\begin{align*}
    \sup_{\substack{\gamma\in\Omega^*\\\lvert \mathrm{I}_-(\gamma)\rvert=n}}E_{\hat{h}}(\gamma)=\sup_{\substack{B\subset\mathbb{Z}\\\lvert B\rvert=n}}\sum_{x\in B}2\hat{h}_x\eqqcolon \bar{E}_{\hat{h}}(n).
\end{align*}

We claim that it is sufficient to prove that
\begin{align}\label{critcampo}
    \frac{\bar{E}_{\hat{h}}(n)}{\underline{H}_J(n)}=\eta<1,
\end{align}
to conclude that the external field with polynomial decay does not destroy the existence of phase transition (in case there is phase transition, of course). Indeed, in this case, 
\begin{align*}
    E_{\hat{h}}(\gamma)\leq \bar{E}_{\hat{h}}(\lvert \mathrm{I}_-(\gamma)\rvert)\leq \eta \underline{H}_J(\lvert \mathrm{I}_-(\gamma)\rvert)\leq \eta H_J(\gamma).
\end{align*}
That is, such a criterion is indeed sufficient to prove that the decaying field does not destroy the existence of a phase transition.
\begin{remark}
    Proving that the same bound as \eqref{critcampo} holds for the non-truncated field is also sufficient.
\end{remark}

\subsection{External fields with polynomial decay}\label{decaying_field}

Let us now consider the case where $J_{xy}=\lvert x-y\rvert^{-\alpha}$, with $1<\alpha\leq 2$, and
\begin{align}\label{decaying}
    h_x=\begin{cases}
        h_*, &\textrm{ if }x=0,\\
        h_*\lvert x\rvert^{-\delta},  &\textrm{ if }x\neq 0.
    \end{cases}
\end{align}

The study of models with decaying fields started in \cite{Bissacot2010}, with non-trivial results in \cite{Bissacot_2015}. This class of models is interesting since we can not use the pressure function as in some arguments for models with interactions that are translation invariant. Since $\mathbb{Z}^d$ is amenable, all the models with fields as in \ref{decaying} have the same pressure function as the model with zero field; however, if the field decays slowly, there is no phase transition \cite{Bissacot_2015, CV}. There are models with identical pressure functions but different behavior concerning the number of Gibbs measures. In this section, we obtain what seems to be a sharp region for the exponents (interaction and decaying field) to obtain phase transition, extending a previous result from \cite{Bissacot-Kimura2018}. Other papers dealing with decaying fields and long-range interactions are \cite{Affonso.2021, potts}.

\begin{lemma}\label{campos}
    The following two estimates for the zero-field energy hold
    \begin{align}
        & \underline{H}_J(n)\geq 4+\frac{4}{(\alpha-1)(2-\alpha)}\left[(n+1)^{2-\alpha}-2^{2-\alpha}\right],&\textrm{ if }1<\alpha<2,\\
        & \underline{H}_J(n)\geq 4+4\log(1+n)-4\log2,&\textrm{ if }\alpha=2.
    \end{align}
    As for the decaying field, we obtain that
    \begin{align}
        & \bar{E}_h(n)\leq h_*\left(3+\frac{2}{1-\delta}n^{1-\delta}\right), &\textrm{ if }0<\delta<1,\\
        & \bar{E}_{\hat{h}}(n)\leq 2h_*\left(\frac{1}{(R+1)^\delta}+\frac{1}{1-\delta}\left[(R+n)^{1-\delta}-(R+1)^{1-\delta}\right]\right),&\textrm{ if }0<\delta<1,\\
        & \bar{E}_h(n) \leq h_*(3+2\log n), &\textrm{ if }\delta=1.
    \end{align}
\end{lemma}

\begin{proof}
    Let us first deal with the Hamiltonian estimates. If $1<\alpha<2$, we have that
    \begin{align*}
        \underline{H}_J(n)&=2\sum_{x=1}^n\left[\sum_{y=-\infty}^0 (x-y)^{-\alpha}+\sum_{y=n+1}^\infty(y-x)^{-\alpha}\right]\\
        &\geq4+2\sum_{x=1}^n\left[\sum_{y=-\infty}^{-1} (x-y)^{-\alpha}+\sum_{y=n+2}^\infty(y-x)^{-\alpha}\right]\\
        &\geq 4+ 2\sum_{x=1}^n\int_{-\infty}^{-1}(x-y)^{-\alpha}\mathrm{d}y+2\sum_{x=1}^n\int_{n+2}^\infty(y-x)^{-\alpha}\mathrm{d}y\\
        &=4+\frac{2}{\alpha-1}\sum_{x=1}^n\left[(x+1)^{1-\alpha}+(n+2-x)^{1-\alpha}\right]\\
        &\geq 4+\frac{4}{\alpha-1}\int_1^n (x+1)^{1-\alpha}\mathrm{d}x=4+\frac{4}{(\alpha-1)(2-\alpha)}\left[(n+1)^{2-\alpha}-2^{2-\alpha}\right].
    \end{align*}
    If $\alpha=2$, the last equality gives us $\underline{H}_J(n)\geq 4+4\log(1+n)-4\log2$.
    
    As for the non-truncated field, we obtain
    \begin{align*}
        \bar{E}_h(n)&\leq h_*\left(1+2\sum_{x=1}^n x^{-\delta}\right)\leq h_*\left(3+2\int_1^n x^{-\delta}\mathrm{d}x\right)\\
        &\leq  h_*\left(3+2\mathbf{1}_{\delta=1}\log n+\frac{2}{1-\delta}\mathbf{1}_{\delta\neq 1}n^{1-\delta}\right).
    \end{align*}
    The estimates for the truncated field are proved in a similar way.
\end{proof}

To be able to prove that phase transition to the region $\delta>\alpha-1$, we need the following auxiliary result.

\begin{lemma}\label{expoentes}
    Let $f:[1,\infty)^2\rightarrow\mathbb{R}$ be given by
    \begin{align*}
        f(x,y)=Kx^\nu-[(x+y)^\mu-(y+1)^\mu],
    \end{align*}
    where $K>0$ and $0<\mu<\nu<1$. Then, there exists $\bar{y}$ such that $f(x,y)\geq 0$ if $y>\bar{y}$.
\end{lemma}

\begin{proof}
We claim that $f(x,y)> 0$ if 
\begin{align}
    y> \bar{y}\coloneqq\left(\frac{\mu}{K\nu}\right)^{\frac{1}{\nu-\mu}}\left[\left(\frac{1-\nu}{1-\mu}\right)^{\frac{1-\nu}{\nu-\mu}}-\left(\frac{1-\nu}{1-\mu}\right)^{\frac{1-\mu}{\nu-\mu}}\right].
\end{align}
Indeed, since $f(1,y)=K>0$ and $\lim_{x\rightarrow\infty}f(x,y)=\infty$, for all $y\in [1,\infty)$. It is sufficient to show that $\partial_xf(x,y)\neq 0$ if $y\geq\bar{y}$. This is a simple calculus problem, one obtains that
\begin{align*}
    \partial_xf(x,y)=0\iff y=-x+\left(\frac{\mu}{K\nu}\right)^{\frac{1}{1-\mu}}x^{\frac{1-\nu}{1-\mu}}\eqqcolon g(x).
\end{align*}
Noting that $0<1-\nu<1-\mu$, we have that $\lim_{x\rightarrow\infty}g(x)=-\infty$. Therefore, it is sufficient to take $y>\bar{y}=\max_{x\geq 1}g(x)$. It is, once again, a calculus exercise to check that
\begin{align*}
    g\left[\left(\frac{\mu}{\nu K}\right)^{\frac{1}{\nu-\mu}}\left(\frac{1-\nu}{1-\mu}\right)^{\frac{1-\nu}{1-\mu}}\right]=\max_{x\geq 1}g(x)=\left(\frac{\mu}{\nu K}\right)^{\frac{1}{\nu-\mu}}\left[\left(\frac{1-\nu}{1-\mu}\right)^{\frac{1-\nu}{\nu-\mu}}-\left(\frac{1-\nu}{1-\mu}\right)^{\frac{1-\mu}{\nu-\mu}}\right].
\end{align*}
The lemma has, thus, been proved.  
\end{proof}

\begin{theorem}
    If $\delta>\alpha-1$, there is phase transition. If $\delta=\alpha-1$, there is phase transition if $h_*$ is sufficiently small.
\end{theorem}

\begin{proof}
    We begin by dealing with the critical exponent. In this case, there is no need to truncate the field. Let $1<\alpha<2$ and $\delta=\alpha-1$. By \eqref{critcampo} and Lemma \eqref{campos}, we need to show that for $h_*$ sufficiently small, we have that
    \begin{align*}
        h_*\left(3+\frac{2}{2-\alpha}n^{2-\alpha}\right)\leq \eta\left\{4+\frac{4}{(\alpha-1)(2-\alpha
        )}[(n+1)^{2-\alpha}-2^{2-\alpha}]\right\},
    \end{align*}
    for all $n\in\mathbb{N}$ and some $0<\eta<1$. We claim this is the case if
    \begin{align}\label{critexponent}
        h_*<\frac{8-4\alpha}{2^{3-\alpha}+6-3\alpha}\eqqcolon 2K_\alpha,
    \end{align}
    in which case the criterion is true for $\eta=h_*(2K_\alpha)^{-1}$. Indeed, the original inequality is equivalent to 
    \begin{align*}
        3h_*-4\eta+\frac{4\cdot 2^{2-\alpha}\eta}{(\alpha-1)(2-\alpha)}\leq \left[\frac{4\eta}{(\alpha-1)(2-\alpha)}-\frac{2h_*}{2-\alpha}\left(\frac{n}{n+1}\right)^{2-\alpha}\right](n+1)^{2-\alpha}\eqqcolon f(n).
    \end{align*}
    On the other hand,
    \begin{align*}
        f(n)\geq \left[\frac{4\eta-2h_*(\alpha-1)}{(\alpha-1)(2-\alpha)}\right](n+1)^{2-\alpha}\geq \left[\frac{4\eta-2h_*(\alpha-1)}{(\alpha-1)(2-\alpha)}\right]2^{2-\alpha}.
    \end{align*}
    It is, therefore, sufficient that
    \begin{align*}
        3h_*-4\eta+\frac{4\cdot2^{2-\alpha}\eta}{(\alpha-1)(2-\alpha)}\leq \left[\frac{4\eta-2h_*(\alpha-1)}{(\alpha-1)(2-\alpha)}\right]2^{2-\alpha},
    \end{align*}
    for some $0<\eta<1$. This gives us \eqref{critexponent}.
    
    Given the computations with the critical exponent, we have that
    \begin{align*}
        K_\alpha\left(3+\frac{2}{2-\alpha}n^{2-\alpha}\right)\leq\frac{1}{2}\left(4+\frac{4}{(\alpha-1)(2-\alpha)}\left[(n+1)^{2-\alpha}-2^{2-\alpha}\right]\right)\leq\frac{1}{2}\underline{H}_J(n),
    \end{align*}
    for all $n\geq 1$. 
    
    Therefore, by \eqref{critcampo} and Lemma \eqref{campos}, to deal with the non-summable subcritical exponents $\alpha-1<\delta<1$, it is sufficient to find $R=R(\alpha,\delta,h_*)$ such that
    \begin{align*}
        2h_*\left\{\frac{1}{(R+1)^{\delta}}+\frac{1}{1-\delta}\left[(R+n)^{1-\delta}-(R+1)^{1-\delta}\right]\right\}\leq K_\alpha\left(3+\frac{2}{2-\alpha}n^{2-\alpha}\right).
    \end{align*}
    It is clearly sufficient to take $R$ sufficiently large so that
    \begin{align*}
        &2h_*(R+1)^{-\delta}\leq 3K_\alpha\textrm{ and}\\
        &(R+n)^{1-\delta}-(R+1)^{1-\delta}\leq \frac{K_\alpha(1-\delta)}{h_*(2-\alpha)}n^{2-\alpha}.
    \end{align*}
    The first one is satisfied if $R\geq\left(\frac{2h_*}{3K_\alpha}\right)^{\frac{1}{\delta}}\eqqcolon R_1(\alpha,\delta,h_*)$. As for the second one, by Lemma \eqref{expoentes}, with $K=\frac{K_\alpha(1-\delta)}{h_*(2-\alpha)}$, $\mu=1-\delta$ and $\nu=2-\alpha$, it is sufficient to take
    \begin{align*}
        R>\left[\left(\frac{\alpha-1}{\delta}\right)^{\frac{\alpha-1}{\delta-(\alpha-1)}}-\left(\frac{\alpha-1}{\delta}\right)^{\frac{\delta}{\delta-(\alpha-1)}}\right]\left(\frac{h_*}{K_\alpha}\right)^{\frac{1}{\delta-(\alpha-1)}}\eqqcolon R_2(\alpha,\delta,h_*).
    \end{align*}
    In which case, phase transition for sub-critical exponents may proved by taking $R=2\max\{R_1,R_2\}$.
    \begin{remark}
        If $1<\alpha<2$, there is no need to treat $\delta=1$ separately. Indeed, let $\delta'=\frac{\alpha}{2}$. Note that $\alpha-1<\delta'<1$. For any $h_*>0$, we have that $h_*\lvert x\rvert^{-1}< h_*\lvert x\rvert^{-\delta'}$, for all $x\neq 0$. By criterion \eqref{critcampo}, there is phase transition for $\delta=1$.
    \end{remark}

    Finally, let $\alpha=2$, the only interesting exponent to analyze is precisely the critical one $\delta=1$. We claim there is phase transition if
    \begin{align}
        h_*<\frac{4}{3+2\log 2}.
    \end{align}
    The reasoning is very similar to what was a done in the critical case for $1<\alpha<2$. The theorem has, thus, been proved.
\end{proof}

\section{Concluding Remarks}

Since Affonso, Bissacot, Endo, and Handa were able to prove \cite{Affonso.2021} (see also \cite{Johanes}) via a multidimensional adaptation of the multiscaled contours, introduced by Fr\"ohlich and Spencer in \cite{FS81} and \cite{Frohlich.Spencer.82}, the existence of phase transition for $\alpha>d$ ($d\geq 2$), it became of interest to revisit the one-dimensional case. This is particularly the case due to how widespread the hypothesis that \textit{$J(1)$ be sufficiently large} is in the one-dimensional case. One of our objectives was precisely to show that such a hypothesis can be relaxed, at least to obtain a proof via contours of phase transition. Still, it is natural to expect that we can remove this condition from all the subsequent papers after \cite{Cassandro.05}, for instance, the results from \cite{Cassandro.Merola.Picco.17, Cassandro.Merola.Picco.Rozikov.14,Cassandro.Picco.09} should be true when we assume $J(1)=1$, and no further restrictions on $\alpha$, except for those relevant to the particular problem at hand. The proof of uniqueness in the case of models with decaying fields is not standard; see \cite{Bissacot_2015, CV}. It remains an open problem to prove uniqueness for models with long-range interactions and decaying fields.

\section*{Acknowledgements}

This study was financed, in part, by the S\~{a}o Paulo Research Foundation (FAPESP), Brazil. Process Number 2023/00854-0. RB was supported by CNPq grant 408851/2018-0; HC and KW are supported by CAPES and CNPq grant 160295/2024-6. RB and HC were partially supported by USP-COFECUB Uc Ma 176/19, "{\it Formalisme Thermodynamique des quasi-cristaux \`a temp\'erature z\'ero}". The authors thank Pierre Picco for many interesting discussions about the 1d contours in Marseille and in S\~ao Paulo; his motivation about the topic inspired us.

\bibliographystyle{habbrv}
\bibliography{refs}

\end{document}